\documentclass[
a4paper
]{article}

\usepackage{authblk}
\usepackage{fullpage}
\usepackage[utf8]{inputenc}
\usepackage{amsmath}
\usepackage{amssymb}
\usepackage{amsthm}
\usepackage{mathtools}
\usepackage{fullpage}
\usepackage[labelfont=bf]{caption}
\usepackage{authblk}

\newif\ifPods
\Podstrue

\usepackage[utf8]{inputenc}
\usepackage{xspace}
\usepackage{graphicx}
\usepackage{MnSymbol}
\usepackage[linesnumbered,ruled,vlined]{algorithm2e}
\usepackage{subcaption}
\usepackage{caption}
\usepackage{xfrac}
\usepackage{bbm}
\usepackage{enumitem}
\usepackage{nicefrac}
\usepackage{thm-restate}

\usepackage{tikz}
\usetikzlibrary{automata,angles,calc,graphs,positioning,chains,arrows,decorations.pathmorphing}

\theoremstyle{definition}

\usepackage
[disable]
{todonotes}

\newcommand{\reinhard}[1]{\todo[inline,color=green!20]{{\bf R:} #1}}

\newcommand{\cristian}[1]{\todo[inline, color=red!20]{{\bf Cr:} #1}}

\DeclareMathOperator*{\argmax}{arg\,max}

\newcommand{\bbN}{\mathbb{N}}
\newcommand{\bbR}{\mathbb{R}}

\newcommand{\bbRgeqz}{\bbR_{\geq 0}}
\newcommand{\bbRgeqzInf}{\bbR_{\geq 0} \cup \{\infty\}}

\newcommand{\sem}[1]{{\lsem{}{#1}\rsem}}
\newcommand{\CQEval}{\textsc{CQEval}}
\newcommand{\CQNext}{\textsc{CQNext}}
\newcommand{\fhw}{\operatorname{fhw}}

\newcommand{\NP}{\textsc{NP}}
\newcommand{\PTIME}{\textsc{P}}

\newcommand{\cS}{\mathcal{S}}
\newcommand{\cV}{\mathcal{V}}
\newcommand{\cB}{\mathcal{B}}

\newcommand{\TUPLES}{\mathbb{T}_\Sigma}

\newcommand{\mucount}{\mu_{\operatorname{count}}}
\newcommand{\muw}{\mu_w}

\newcommand{\U}{\mathcal{U}}
\newcommand{\setsU}{\finite(\U)}
\newcommand{\deltasum}{\delta_{\operatorname{sum}}}
\newcommand{\deltamin}{\delta_{\operatorname{min}}}
\newcommand{\deltaW}{\delta_{\operatorname{W}}}

\newcommand{\nop}[1]{}

\newcommand{\finite}{\operatorname{finite}}
\newcommand{\maxCoverage}{\mathtt{Max\,k{-}Coverage}}
\newcommand{\IndependentSet}{\mathtt{IndependentSet}}
\newcommand{\HamiltonianPath}{\mathtt{HamiltonianPath}}

\newcommand{\Velem}{\cV_{\operatorname{elem}}}
\newcommand{\betaelem}{\beta_{\operatorname{elem}}}

\newcommand{\Vpos}{\cV_{\operatorname{pos}}}
\newcommand{\betapos}{\beta_{\operatorname{pos}}}

\newcommand{\Velemw}{\cV_{\operatorname{elem}}^w}
\newcommand{\Vposw}{\cV_{\operatorname{pos}}^w}

\newcommand{\Vr}{\cV_{r}}
\newcommand{\betar}{\beta_{r}}

\newcommand{\Vg}{\cV_{g}}
\newcommand{\betag}{\beta_{g}}

\newcommand{\VQD}{\cV_{Q,D}}
\newcommand{\betaQD}{\beta_{Q,D}}

\newcommand{\by}{\bar{y}}
\newcommand{\bz}{\bar{z}}

\theoremstyle{plain}
\newtheorem{theorem}{Theorem}[section]
\newtheorem*{theorem*}{Theorem}
\newtheorem{proposition}[theorem]{Proposition}
\newtheorem*{proposition*}{Proposition}
\newtheorem{lemma}[theorem]{Lemma}
\newtheorem*{lemma*}{Lemma}

\newtheorem*{corollary*}{Corollary}

\newtheorem*{claim*}{Claim}

\newtheorem*{conjecture*}{Conjecture}

\theoremstyle{definition}
\newtheorem{definition}[theorem]{Definition}
\newtheorem*{definition*}{Definition}

\theoremstyle{remark}

\newtheorem*{observation*}{Observation}
\newtheorem{example}[theorem]{Example}
\newtheorem*{example*}{Example}

\newif\ifArxiv
\Arxivtrue

\begin{document}

\title{Query Answering Under Volume-Based Diversity Functions}

\author[1]{Marcelo Arenas}
\author[2]{Timo Camillo Merkl}
\author[2]{Reinhard Pichler}
\author[1]{Cristian Riveros}
\affil[1]{Pontificia Universidad Católica de Chile, \texttt{marenas@uc.cl}, \texttt{cristian.riveros@uc.cl}}
\affil[2]{TU Wien, Austria, \texttt{timo.merkl@tuwien.ac.at}, \texttt{pichler@dbai.tuwien.ac.at}}

\date{} 

\maketitle

\begin{abstract}
When query evaluation produces too many tuples, a new approach in query answering is to retrieve a diverse subset of them. The standard approach for measuring the diversity of a set of tuples is to use a distance function between tuples, which measures the dissimilarity between them, to then aggregate the pairwise distances of the set into a score (e.g., by using sum or min aggregation). However, as we will point out in this work, the resulting diversity measures 
may display some unintuitive behavior. Moreover, even in very simple settings, 
finding a maximally diverse subset of the answers of fixed size is, in general, intractable and little is known about approximations apart from some hand-picked distance-aggregator pairs. 

In this work, we introduce a novel approach for computing the diversity of tuples based on volume instead of distance. We present a framework for defining volume-based diversity functions and provide several examples of these measures applied to relational data. Although query answering of conjunctive queries (CQ) under this setting is intractable in general, we show that one can always compute a (1-1/e)-approximation for any volume-based diversity function. Furthermore, in terms of combined complexity, we connect the evaluation of CQs under volume-based diversity functions with the ranked enumeration of solutions, finding general conditions under which a (1-1/e)-approximation can be computed in polynomial time. 
\end{abstract}

	\section{Introduction}\label{sec:introduction}

When the set of answers to a query gets too big, a user might be better served by
being  presented a meaningful
subset of the answers rather than being overwhelmed with the entire set. 
Clearly, sampling might provide one way of selecting a ``representative'' subset of the answers.
However, as was pointed out in \cite{DBLP:conf/icde/VeeSSBA08}, such an approach 
typically misses interesting but rarely occurring answers.
An alternative approach, which has recently received increased attention by the database community,
is  to aim at a small, {\em diverse} set of answers~\cite{DBLP:journals/pacmmod/AgarwalEHSY24,%
DBLP:journals/corr/abs-2408-01657,%
DBLP:journals/vldb/IslamAAR23,%
MerklPS23,%
DBLP:journals/vldb/NikookarEBSAR23}.
For instance (following an example given in \cite{DBLP:conf/aaai/HebrardHOW05}), in a car dealership setting, 
the number of models satisfying the constraints expressed by the customer may be huge. 
Therefore, rather than presenting all solutions to this constraint satisfaction problem (which is a well-known
equivalent problem to conjunctive query answering \cite{DBLP:journals/jcss/KolaitisV00}),
it would be more useful to come up with a small, {\em diverse} set of solutions and let the 
customer decide on which type of models to focus further  discussions.

The most common approach of assigning a diversity score $\delta$ to a subset of the universe $\U$ 
(e.g., the answers to a query) is to first  
define a distance measure $d$ between any two distinct elements of $\U$ and then 
define the diversity $\delta(S)$ of any subset $S$ of $\U$ by applying some 
aggregation to the pairwise distances of the elements in $S$~\cite{DBLP:conf/aaai/IngmarBST20}. 
Typical 
distance measures in the database context 
are the Hamming distance \cite{DBLP:journals/pacmmod/AgarwalEHSY24, DBLP:journals/ai/BasteFJMOPR22,MerklPS23} (i.e., counting
the positions in which two tuples differ), an ultrametric \cite{DBLP:conf/icde/VeeSSBA08,DBLP:journals/debu/VeeSA09}
(i.e., imposing an order on the attributes and considering tuples farther apart if they differ
on an attribute further up in this order), or the Euclidean distance for numeric 
attributes~\cite{DBLP:journals/pacmmod/AgarwalEHSY24}. Typical aggregate functions
are the sum and min operators~\cite{DBLP:journals/pacmmod/AgarwalEHSY24, DBLP:journals/ai/BasteFJMOPR22,MerklPS23}.

While sum and min are natural and familiar aggregate functions, they may lead to some anomalies of the 
resulting diversity measure: In case of the sum operator, consider a setting where a set $S$ contains two elements $t_1, t_2$ with high distance
$d(t_1, t_2)$. Then adding to $S$ 
another element $t_1'$ very close to $t_1$ but again with high distance from $t_2$ 
seemingly leads to a significant increase of diversity, even though $t_1'$ is almost a ``copy'' of $t_1$.
In \cite{weitzman1992diversity} several desiderata on diversity measures are presented -- 
including the {\em twin property}, i.e., adding an (almost) identical copy should not increase the diversity, and {\em monotonicity},
i.e., adding a new element to a set $S$ never decreases the diversity. Clearly, the sum operator violates the first fundamental property while the min operator violates this second property.
Consequently, Weitzman~\cite{weitzman1992diversity} introduced a diversity measure (henceforth referred to as $\deltaW$)
based on a more sophisticated aggregation of pairwise distances. However, as was 
shown in \cite{DBLP:journals/corr/abs-2408-01657}, even the basic task of determining the diversity $\deltaW(S)$ of a given set
$S$ of elements is NP-complete. 

The goal of our work is to introduce a novel framework for defining diversity measures, such that this framework
is generally applicable but, at the same time, 
particularly well suited 
for defining natural diversity measures in the (relational) database world. We will thus introduce a two-staged approach which, in the first place,
assigns to each element of the universe (e.g., a tuple in a relation or in an entire database) a {\em volume} in the form of some measurable set.
As will be illustrated in Section \ref{sec:volume}, for a set $S$ of tuples, there are many ways of choosing such a volume. In the simplest
case, we could just collect the set of values occurring in $S$. Various other options,  such as considering $k$-ary balls of a pre-specified radius $r$ 
around a $k$-tuple of numerical attributes are presented in Section \ref{sec:volume}. The second stage then consists in assigning values to the unions of
these measurable sets. For the basic case of collecting the set $V$ of values occurring in $S$, we could simply take the cardinality of $V$. 
For the case of $k$-ary balls associated with each tuple, we would take the volume of the union of the balls associated with the 
tuples in $S$.  A formal definition of our {\em volume-based} approach to diversity will be given in Section \ref{sec:volume}.

We will then  study interesting properties of this approach. In particular, we will analyze its relationship with 
previous approaches -- in particular, Weitzman's approach~\cite{weitzman1992diversity} 
and the multi-attribute approach of Nehring and Puppe~\cite{nehring2002diversitytheory} 
(in Section~\ref{sec:MultiAttribute}, we will formally define that approach and also point at its
major shortcoming, namely the conceptual and computational complexity caused 
by having to deal with the powerset of the powerset of the universe).
Somewhat surprisingly, we will show that diversity functions defined via the multi-attribute approach can also be defined in our framework and, for a finite universe, also the converse holds.
In other words, while avoiding the negative computational properties
of the multi-attribute approach, our volume-based diversity measures share the favorable properties shown in \cite{nehring2002diversitytheory}. 
One of them is {\em submodularity}, which formalizes the 
intuition that  adding a new element to a smaller set potentially leads to a bigger increase of diversity than adding the same element to a bigger set. 

When analyzing computational properties of volume-based diversity measures, submodularity will prove beneficial. Concretely, we study the problem of 
searching for a subset of the answers to a conjunctive query which, 
for a given size $k$, maximizes the diversity. This problem has been studied before for various
diversity functions and, even in very simple settings (e.g., considering the sum or the minimum of the pairwise Hamming distances) this problem
was shown to be intractable~\cite{MerklPS23} -- even for data complexity. 
We will show that also for the natural volume-based diversity measures of sets of tuples presented in 
Section~\ref{sec:volume}, intractability holds. We therefore study the search for a maximally diverse set of answers to a conjunctive query 
from an approximation point of view. For data complexity, we prove a tractable $(1-\nicefrac{1}{e})$-approximation (where $e$ is the Euler number) 
of the maximum diversity score 
for arbitrary volume-based diversity measures by making use of 
{a classical approximation result for submodular set functions
by Nemhauser et al.~\cite{nemhauser1978analysis}}. 
We also show that, in general,
a better tractable approximation can be excluded unless P = NP. 
Clearly, combined complexity requires further restrictions since even query evaluation of Boolean conjunctive queries (without paying any attention to diversity)
is NP-complete~\cite{DBLP:conf/stoc/ChandraM77}. However, by restricting our attention to CQs of 
bounded fractional hypertreewidth~\cite{GroheM14} and establishing a 
relationship with ranked enumeration~\cite{deep2025ranked}, we manage to achieve tractable $(1-\nicefrac{1}{e})$-approximation of the maximum diversity score
also for combined complexity. 

\paragraph{Structure of the paper and summary of results}
After recalling some basic notions in Section~\ref{sec:preliminaries}, we will formally 
introduce our volume-based framework of defining diversity functions in Section~\ref{sec:volume}. 
By presenting some examples of natural diversity functions for 
sets of tuples, we illustrate the suitability of this framework in the database context. 
We then study the relationship of our volume-based approach of defining diversity measures with previous approaches, namely 
with the multi-attribute approach of \cite{nehring2002diversitytheory} in Section~\ref{sec:MultiAttribute} and 
with distance-based approaches (above all Weitzman's diversity measure~\cite{weitzman1992diversity})
in Section~\ref{sec:DiversityDistance}.
The search for a maximally diverse set of $k$ answers to a conjunctive query $Q$ over a given database $D$ is studied in 
Sections~\ref{sec:Exact} and \ref{sec:ACQs}. As mentioned above, a tractable exact solution to this maximization problem is out of reach. We therefore 
settle for an approximation. In Section~\ref{sec:Exact}, we study data complexity and establish a tractable $(1-\nicefrac{1}{e})$-approximation of the maximum diversity score of $k$-element subsets of the answers to first-order queries by virtue of the submodularity of volume-based diversity measures. 
In Section~\ref{sec:ACQs}, we study combined complexity and identify a sufficient condition on the queries to achieve the same quality of approximation.
We conclude with Section~\ref{sec:conclusions}.
\ifArxiv 
Most proofs are deferred to the appendix.

\else
Proof details can be found in the full version of this paper \cite{ARXIV}.

\fi 	
	\section{Preliminaries}\label{sec:preliminaries}

\newcommand{\D}{\mathbb{D}}
\newcommand{\arity}{\operatorname{arity}}
\newcommand{\relnames}{\mathcal{R}}
\newcommand{\Var}{\mathcal{X}}
\newcommand{\ba}{\bar{a}}
\newcommand{\bb}{\bar{b}}
\newcommand{\bc}{\bar{c}}
\newcommand{\bx}{\bar{x}}

\paragraph{Sets and sequences} 
We denote by $\bbN$, $\bbR$, and $\bbRgeqz$  the set of natural,  real and non-negative real numbers, respectively.
Given a set $A$, we denote by $\finite(A)$ the set of all non-empty finite subsets of $A$. For $k \in \bbN$, we say that $B \in \finite(A)$ is a \emph{$k$-subset} if $|B| = k$.
We usually use $a$, $b$, or $c$ to denote elements, and $\ba$, $\bb$, or $\bc$ to denote sequences of such elements. For $\ba = a_1, \ldots, a_{k}$, we write $\ba[i] := a_i$ to denote the $i$-th element of~$\ba$ and $|\ba| := k$ to denote the length of $\ba$. Further, given a function $f$ we write $f(\bar{a}) := f(a_1), \ldots, f(a_{k})$ to denote the function applied to each element of $\ba$.

\paragraph{Conjunctive queries} Fix a set $\D$ of data values. A relational schema $\Sigma$ (or just schema) is a pair $(\relnames, \arity)$,
where $\relnames$ is a set of relation names and $\arity: \relnames \rightarrow \bbN$ assigns 
each name to a number. An $R$-tuple of $\Sigma$ (or just a tuple) is a syntactic object
$R(a_1, \ldots, a_{k})$ such that $R \in \relnames$, $a_i \in \D$ for every~$i$,  
and $k = \arity(R)$. We will write $R(\ba)$ to denote a tuple with values $\ba$. 
Given a schema $\Sigma$, we denote by $\TUPLES$ the set of all tuples over $\Sigma$ with values in  $\D$.
A \emph{relational database} $D$ over $\Sigma$ is a finite set of tuples over $\Sigma$.
For a schema $\Sigma = (\relnames, \arity)$ and a set of variables $\Var$ disjoint from $\D$, a \emph{Conjunctive Query} (CQ) over $\Sigma$ is a syntactic structure of the~form:
$$
Q(\bx) \ \leftarrow \ R_1(\bx_1), \ldots, R_{m}(\bx_{m}) 
$$
such that $Q$ denotes the answer relation, each $R_i$ is a relation name in $\relnames$, $\bx_i$ is a sequence of variables in $\Var$, 
$|\bx| = \arity(Q)$, and $|\bx_i| = \arity(R_i)$ for every $i \leq m$. Further, $\bx$ is a sequence of variables appearing in $\bx_1, \ldots, \bx_{m}$. 
We refer to such a CQ simply as $Q$, where $Q(\bx)$ and $R_1(\bx_1), \ldots, R_{m}(\bx_{m})$ are called the \emph{head} and the \emph{body} of $Q$, respectively. 
Furthermore, we call each $R_i(\bx_i)$ an \emph{atom} of $Q$, and we say that $Q$ is a \emph{full} CQ if each variable occurring in the body of $Q$ also appears in the head of $Q$.

Let $Q$ be a CQ of the above form, and $D$ be a database over the same schema~$\Sigma$. A \emph{homomorphism} from $Q$ to $D$ is a function $h: \Var \rightarrow \D$ such that $R_i(h(\bx_i)) \in D$ for every $i \leq m$. We define the \textit{answers} of $Q$ over $D$ as the set of $Q$-tuples
$
\sem{Q}(D) \ := \ \{Q(h(\bx)) \mid \text{$h$ is a homomorphism from $Q$ to $D$}\}.
$

\paragraph{Distance-based diversity} Let $\U$ be an infinite set.
We see $\U$ as a \emph{universe} of possible solutions and $S \in \setsU$ as a candidate finite set of solutions.
In its most general form, a diversity function over $\U$ is a function $\delta \colon \setsU \to \bbRgeqzInf$. 
The standard approach, that we call here \emph{distance-based diversity} functions, 
is to first define a {\em distance function}  $d \colon \U \times \U \rightarrow \bbRgeqzInf$ 
(typically $d$ is a metric on $\U$)
and to define the diversity $\delta$ as an extension of  $d$ from pairs to arbitrary subsets of  $\U$ 
setting $\delta(S) = 0$ if $|S| \leq 1$.
As proposed in \cite{DBLP:conf/aaai/IngmarBST20}, 
one way of defining $\delta$ for a given 
distance function $d$ is to define an aggregator $f$ that combines the pairwise distances. That is, 
we set 
$\delta(S): = f(d(a,b)_{a,b \in S})$.
The most common aggregators are sum and min, which give rise to the following diversity functions: 
\[
\deltasum(S) := \sum_{a,b \in S} d(a,b) \ \ \ \ \text{ and }  \ \ \ \  \deltamin(S) := \min_{a,b \in S \,:\, a \neq b} d(a,b).
\]

	\section{Volume-based Diversity Framework}\label{sec:volume}

In this section, we introduce a general volume-based framework for measuring the diversity of sets of tuples. We begin by recalling the definitions of $\sigma$-algebra and measures. Then, we introduce the main definitions of the framework, present several examples to motivate the use of volume-based diversity measures over relational data, and prove that volume-based diversity functions satisfy two fundamental properties expected of diversity measures.

\paragraph{Measures} We recall here the standard definitions of $\sigma$-algebra and measures (see e.g.~\cite{axler2020measure} for further details).
Let $\Omega$ be a set (possibly infinite).
A \emph{$\sigma$-algebra} over $\Omega$ is a family $\cS$ of subsets of $\Omega$ (i.e., $\cS \subseteq 2^{\Omega}$) such that (1) $\emptyset \in \cS$, 
(2) if $X \in \cS$, then $\Omega \setminus X \in \cS$, and (3) if $X_i \in \cS$ for every $i \in \mathbb{N}$, then $\bigcup_{i \in \mathbb{N}} X_i \in \cS$.
Given a $\sigma$-algebra $\cS$, a \emph{measure} for $\cS$ is a function $\mu : \cS \to \bbRgeqzInf$ such that (1)~$\mu(\emptyset) = 0$ and (2) if $X_i \in \cS$ for every $i \in \mathbb{N}$ and $X_i \cap X_j = \emptyset$ for every $i, j \in \mathbb{N}$ with $i \neq j$, then: 
\[
		\mu(\bigcup_{i \in \mathbb{N}} X_i) \ = \ \sum_{i \in \mathbb{N}} \mu(X_i).
\]
For example, assuming that $\Omega$ is a countable set, one can check that $2^{\Omega}$ is a $\sigma$-algebra and $\mucount: 2^{\Omega} \rightarrow \bbRgeqzInf$ that maps $\mucount(X) = |X|$ if $X$ is finite and $\mu_0(X) = \infty$, otherwise, is a measure, called the \emph{counting measure}. Another example is the {\em weighted measure}, where we 
consider a weight function $w: \Omega \rightarrow \bbRgeqz$ over $\Omega$ and define $\mu_w(X) = \sum_{a \in X} w(a)$ where the sum is defined as the supremum of $\sum_{a \in Y} w(a)$ over all finite subsets $Y \subseteq X$.  A particular case here is a probability distribution over a $\sigma$-algebra $\cS$ where $\mu$ assigns a probability in $[0,1]$ to each subset $X$ of $\Omega$.

\paragraph{The volume-based framework} Assume that $\U$ is the \emph{universe of possible solutions} over which we want to measure diversity. A \emph{volume assignment} $\cV$ over $\U$ is a tuple $\cV =(\cS, \mu, \beta)$ such that $\cS$ is a $\sigma$-algebra over a set $\Omega$ (that may be different from $\U$), $\mu$ is a measure for $\cS$ and $\beta: \U \to \cS$. 
Intuitively, the function $\beta$, called the \emph{ball function}, is a function that assigns a \textit{ball} in $\cS$ to each element $a$ of the universe $\U$, namely, it assigns a volume to $a$.

We now introduce our framework for defining diversity functions over volume assignments as follows.
Given a volume assignment $\cV = (\cS, \mu, \beta)$ over $\U$, a function $\delta_{\cV}: \setsU \to \bbRgeqzInf$ is a \textit{volume-based diversity function} over $\U$ if for every $S \in \setsU$:
\[
\delta_{\cV}(S) \ = \ \mu(\bigcup_{a \in S} \beta(a)).
\]
Intuitively, each element of $S$ contributes with different characteristics to the diversity of the group (i.e., its volume), and when we add all these characteristics together, the intersection only adds once. In particular, when two elements $a_1$ and $a_2$ of $\U$ are totally different with respect to the diversity (i.e. $\beta(a_1) \cap \beta(a_2) = \emptyset$), we have that $\delta_\cV(\{a_1,a_2\}) = \delta_\cV(\{a_1\}) + \delta_\cV(\{a_2\})$. Also, note that, contrary to distance-based diversity functions (defined in Section \ref{sec:preliminaries}), it is not necessary that $\delta_\cV(\{a\}) = 0$ (indeed, $\delta_\cV(\{a\}) \neq 0$ almost surely). Depending on the application context, positive diversity of singletons might actually be the desired behavior. For instance, suppose that the universe $\U$ denotes the set of employees, $\Omega$ a set of skills, and $\beta$ assigns to each employee her skills. Then 
the diversity 
$\delta(S)$ assigns a measure (via $\mu$) to the skill set present in a team $S$ of employees. In this case, we
clearly want $\delta$ applied to a singleton to reflect the value of the skills possessed by each individual. 

In the following, we provide several examples of volume-based diversity functions applied to relational data (i.e., tuples). For this purpose, recall that we use $\D$ to denote a set of data values, $\Sigma$ to denote an arbitrary relational schema, and $R(a_1, \ldots, a_{k})$ to denote an $R$-tuple of $\Sigma$ where $a_i \in \D$ for every $i$. In the following examples, we use $\TUPLES$ to denote our universe (i.e., $\U$) of all possible~tuples. 

\begin{example} \label{ex:elem-volume}
	Let $\Velem = (2^{\D}, \mucount, \betaelem)$ be the volume assignment such that 
	$\cS = 2^{\D}$ is the $\sigma$-algebra (over $\D$), $\mucount$ is the counting measure 
and $\betaelem$ is the ball function defined as:
	\[
	\betaelem(R(a_1, \ldots, a_{k})) \ := \ \{a_1, \ldots, a_{k}\}
	\]
	for every tuple $R(a_1, \ldots, a_{k})$. For every finite set of tuples $S \subseteq \TUPLES$, we have that $\delta_{\Velem}(S)$ measures the number of different data values contained in the tuples in $S$. That is, the more different data values the tuples have, the more diverse they are.
	
	\begin{figure}[t]
		\centering
		\begin{tabular}{ccccccccc}
			$D_1:$ & \begin{tabular}{|c c|}\hline
				\multicolumn{2}{|c|}{$R$}\\\hline\hline
				$a$ & $a$\\
				$a$ & $b$\\
				$b$ & $a$\\\hline
				\multicolumn{2}{c}{}\\
			\end{tabular}
			&\ \ \ \ \ \ \ \ \ \ \ & $D_2:$ & \begin{tabular}{|c c|}\hline
				\multicolumn{2}{|c|}{$R$}\\\hline\hline
				$a$ & $a$\\
				$a$ & $b$\\
				$b$ & $a$\\
				$b$ & $b$\\\hline
			\end{tabular}
			&\ \ \ \ \ \ \ \ \ \ \ & $D_3:$ & \begin{tabular}{|cc|}\hline
				\multicolumn{2}{|c|}{$R$}\\\hline\hline
				$a$ & $b$\\
				$a$ & $c$\\\hline
				\multicolumn{2}{c}{}\\
				\multicolumn{2}{c}{}\\
			\end{tabular}
		\end{tabular}
		\caption{Databases $D_1$, $D_2$ and $D_3$ consisting of a binary relation $R$ with data values $a$, $b$, and $c$.}
		\label{fig:running-example}
	\end{figure}
    For instance, consider the databases $D_1$, $D_2$ and $D_3$ in Figure~\ref{fig:running-example} consisting of a binary relation $R$.
    We have that $\betaelem(R(a,a)) = \{a\}$, $\betaelem(R(a, b)) = \{a,b\}$,  $\betaelem(R(b, a)) = \{a,b\}$ and
    \begin{align*}
        \delta_{\Velem}(D_1) = \mucount(\betaelem(R(a,a)) \cup \betaelem(R(a,b)) \cup \betaelem(R(b,a))) = \mucount(\{a,b\}) = |\{a,b\}| = 2.
    \end{align*}
    In the same way, we conclude that $\delta_{\Velem}(D_2) = 2$ and $\delta_{\Velem}(D_3) = 3$. Hence, $D_1$ and $D_2$ are equally diverse under the measure $\delta_{\Velem}$, while $D_3$ is considered more diverse than these two databases, as it contains an extra value.
    \qed
\end{example}

\begin{example} \label{ex:pos-volume}
    In addition to measuring the diversity in data values, we now also want to consider the position where these data values occur, i.e., it is different whether $a$ appears in the first or second component of a tuple. We thus capture the intuition that different attributes, even if they
    have the same data type, have a different semantics (e.g., in a car-relation, the number 6 occurring both in the ``gears'' and  in the ``cylinders'' attribute does not reduce the diversity). For this purpose, consider the volume assignment $\Vpos = (2^{\D \times \bbN}, \mucount, \betapos)$ where we use the $\sigma$-algebra $2^{\D \times \bbN}$ (over $\D \times \bbN$), the counting measure $\mucount$ and the ball function $\betapos$ such that:
	\[
	\betapos(R(a_1, \ldots, a_{k})) \ := \ \{(a_1, 1), \ldots, (a_{k}, k)\}
	\]
	for every tuple  $R(a_1, \ldots, a_{k}) \in \TUPLES$. Then, the diversity $\delta_{\Vpos}(S)$ measures the number of different values that appear in different positions of the tuples in $S \subseteq \TUPLES$.
    For instance, 
    consider again the databases $D_1$, $D_2$ and $D_3$ given in Figure~\ref{fig:running-example}. Then we have that $\betapos(R(a,a)) = \{(a,1),(a,2)\}$,  $\betapos(R(a,b)) = \{(a,1),(b,2)\}$, $\betapos(R(b,a)) = \{(b,1),(a,2)\}$ and
    \begin{multline*}
        \delta_{\Vpos}(D_1) = \mucount(\betapos(R(a,a)) \cup \betapos(R(a,b)) \cup \betapos(R(b,a))) =\\ \mucount(\{(a,1), (a,2), (b,2), (b,1)\}) = 4.
    \end{multline*}
In the same way, we conclude that $\delta_{\Vpos}(D_2) = 4$ and $\delta_{\Vpos}(D_3) = 3$. Hence, as opposed to the diversity measurements given in Example \ref{ex:elem-volume}, $D_1$ and $D_2$ are equally diverse under the measure $\delta_{\Vpos}$, while $D_3$ is considered less diverse than these two databases, as it contains a smaller number of values in different positions. 
\qed
\end{example}

\begin{example} \label{ex:weight-volume}
	Another practical example of volume-based diversity functions is considering a weight function $w: \D \rightarrow \bbRgeqz$.
For instance, if the relational data considers animals in $\D$, then a user could use a weight function where  $w(\text{`dog'})$ will weigh less than $w(\text{`dodo'})$ given that dodo is a less common animal than a dog.
	Then one can consider the volume assignment $\Velemw = (2^{\D}, \muw, \betaelem)$ where the $\sigma$-algebra and ball functions are the same as in $\Velem$ (see Example~\ref{ex:elem-volume}) and the measure $\muw$ is the weighted measure defined above.
	Then $\delta_{\Velemw}(S)$ measures the weight of the data values appearing in tuples, assigning more diversity to tuples where a dodo appears versus a dog. 
	
	One can naturally extend this example to also consider the positions of the data values (denoted by $\Vposw$) as in Example~\ref{ex:pos-volume} and, instead of a weight function $w$, one can use a probability function that assigns a probability to each data value. 
    To showcase $\Vposw$, consider again the databases $D_1$, $D_2$ and $D_3$ given in Example \ref{ex:elem-volume}. Moreover, assume that $c$ is an uncommon value for the second attribute of $R$, which is represented by the following weight function: $w((a,1)) = w((a,2)) = w((b,1)) = w((b,2)) = w((c,1)) = 1$, and $w((c,2)) = 3$.
    Then we have that:
    \begin{multline*}
        \delta_{\Vposw}(D_3) = \muw(\betapos(R(a,b)) \cup \betapos(R(a,c))) =\\ \muw(\{(a,1), (b,2), (c,2)\}) = 
        w((a,1)) + w((b,2)) + w((c,2)) = 5.
    \end{multline*}
In the same way, we conclude that $\delta_{\Vposw}(D_1) = 4$ and $\delta_{\Vpos}(D_2) = 4$. Hence, in this case $D_3$ is considered as the most diverse database given the occurrence of $c$ in the second column of $R$. 
	\qed
\end{example}

\begin{example} \label{ex:VQD}
	Let $D$ be a relational database over a schema $\Sigma$ and  consider a CQ $Q(\bx) \ \leftarrow \ R_1(\bx_1), \ldots, R_{m}(\bx_{m})$.
A user may want to measure the diversity of a subset $S \subseteq \sem{Q}(D)$ concerning the provenance of each tuple, namely, which are the tuples in $D$ that contribute to the outputs in $S$ (cf.\ the ``which provenance'' studied in \cite{DBLP:journals/vldb/CuiW03}). One way to formalize this is as follows. Let $\sem{Q}(D)$ be the universe of possible solutions. Consider the volume assignment $\VQD = (2^{D}, \mucount, \betaQD)$ where $2^{D}$ is the $\sigma$-algebra (i.e., all subsets of tuples in $D$), $\mucount$ is the counting measure, and $\betaQD:\sem{Q}(D) \rightarrow 2^{D}$ is the ball function such that for every answer~$Q(\ba)\in\sem{Q}(D)$:
	\[
	\betaQD(Q(\ba)) \ := \ \{R_i(h(\bx_i)) \mid 1 \leq i\leq m \text{ and $h$ is a homomorphism from $Q$ to $D$ with $h(\bx)=\ba$}\}
	\]
	In other words, $\betaQD$ maps $Q(\ba)$ to all the tuples that contribute to it, that is, its provenance. For $S \subseteq \sem{Q}(D)$, the value $\delta_{\VQD}(S)$ counts the number of different tuples in $D$ that support the outputs in $S$. Then, the more different tuples support $S$, the more diverse they are. 
        For instance, 
    consider again the database $D_1$ given in Example \ref{ex:elem-volume}, and let $Q_1(x,y)$ be the conjunctive query $\exists z \, R(x, z) \wedge R(z, y)$. Then the tuples $(a, a)$ and $(b,b)$ are both answers to $Q_1(x, y)$ over $D_1$. However, we have that $\beta_{Q_1,D_1}(Q_1(a,a)) = \{R(a,a), R(a,b), R(b,a)\}$,  $\beta_{Q_1,D_1}(\{Q_1(b,b)\}) = \{R(b,a),R(a,b)\}$~and
    \begin{align*}
        \delta_{\cV_{Q_1,D_1}}(\{Q_1(a,a)\}) =& \ \mucount(\beta_{Q_1,D_1}(Q_1(a,a))) = \mucount(\{R(a,a), R(a,b), R(b,a)\}) = 3,\\
        \delta_{\cV_{Q_1,D_1}}(\{Q_1(b,b)\}) =& \ \mucount(\beta_{Q_1,D_1}(Q_1(b,b))) = \mucount(\{R(b,a), R(a,b)\}) = 2.
    \end{align*}
Hence, in this case, $Q_1(a,a)$ is considered a more diverse answer than $Q_1(b,b)$, as there is a larger number of ways in which $(a,a)$ can be obtained as an answer to $Q_1(x,y)$ over $D_1$. \qed
\end{example}

\begin{example}
 We now consider a more geometrical scenario where $\D = \bbR$ and a tuple $R(a_1, \ldots, a_k)$ represents points in the $\bbR^k$-space. Then, given a radius $r > 0$ we can define the volume assignment $\Vr = (\cB^k, \mu, \betar)$ where $\cB^k$ are the Borel sets of $\bbR^k$ (i.e., measurable sets), $\mu$ is the Lebesgue measure (i.e., measures the volume of a measurable set in $\cB^k$), and $\beta_r$ is the ball function such that:
	\[
	\betar(R(a_1, \ldots, a_{k})) \ := \ \{(b_1, \ldots, b_k) \in \bbR^k \mid \sqrt{(a_1 - b_1)^2 + \ldots +(a_k - b_k)^2} \leq r\}
	\]
	namely, $\beta_r$ assigns a ball of radius $r$ under euclidean distance around $(a_1, \ldots, a_{k})$. Then the volume-based diversity function $\delta_{\Vr}(S)$ measures the volume of $r$-balls around points in $S$. 
	In particular, the farther apart (up to radius $r$) the points in $S$, the more diverse they are. \qed
\end{example}

\begin{example}
As our last example, 
	we adapt the previous example to have points closer to the tuples $R(a_1, \ldots, a_{k})$ contribute more to the diversity than points further away by adding Gaussian functions.
    To that end, again let $\D = \bbR$ and a tuple $R(a_1, \ldots, a_k)$ represents points in the $\bbR^k$-space. 
    Then, we can define the volume assignment $\Vg = (\cB^{k+1}, \mu, \betag)$ where $\cB^{k+1}$ are the Borel sets of $\bbR^{k+1}$ (note that we added a dimension), $\mu$ is the Lebesgue measure, and $\betag$ is the ball function with 
	\[
	\betag(R(a_1, \ldots, a_{k})) \ := \ \{(b_1, \ldots, b_k, d) \in \bbR^{k+1} \mid 0\leq d\leq e^{-(a_1 - b_1)^2 - \ldots -(a_k - b_k)^2}\}
	\]
	namely, $\betag$ assigns the area under the Gaussian function centered around $(a_1, \ldots, a_{k})$. Then the volume-based diversity function $\delta_{\Vg}(S)$ measures the collective volume under the Gaussian functions, i.e., the integral $\int_{\mathbb{R}^k}\max_{s\in S}e^{||x-s||_2^2}\,dx$.
        A benefit of using Gaussian over simple boxes is that adding a new element will always increase the diversity at least a bit.
\end{example}

\paragraph{Monotonicity and submodularity} We conclude this section by 
introducing two fundamental properties of
diversity functions,  
{advocated for in 
\cite{nehring2002diversitytheory,weitzman1992diversity}.}

Fix a universe $\U$ of possible solutions and a volume assignment $\cV$ over $\U$. A
first desirable property of diversity functions is that
of \emph{monotonicity}: adding an element to a set cannot decrease the
diversity of the set. Formally, a diversity function $\delta$ is
monotone if $\delta(S \cup \{a\}) \geq \delta(S)$ for every
$S \in \setsU$ and $a \in \U$.

A second desirable property of diversity functions 
is \emph{submodularity}%
\footnote{{We note that, in contrast to Nehring and Puppe~\cite{nehring2002diversitytheory}, 
Weitzman~\cite{weitzman1992diversity} does not explicitly propose 
{\em submodularity} as a desideratum. However, he mentions that, ideally, 
the increase of diversity when adding a new element $a$ to $\U$, should correspond to the minimum distance of $a$ from the already existing elements in $\U$. 
Clearly, this property implies submodularity.}},
which  means that,  
for every $a \in \U$ and $S_1, S_2 \in \setsU$ with 
$S_1 \subseteq S_2$, the property
$
\delta(S_1 \cup \{a\}) - \delta(S_1) \geq \delta(S_2 \cup \{a\}) - \delta(S_2)
$
holds.
As mentioned in Section~\ref{sec:introduction}, 
submodularity captures
the intuition that adding an element $a$ to the smaller set $S_1$ should
result in a greater
increase in diversity than adding it to $S_2$.%

In the next proposition, we show that both
properties are satisfied by volume-based diversity functions,
thereby
providing evidence of the naturalness of our approach.

\begin{proposition}
\label{prop:mon-subm}
Let $\cV$ be any volume assignment over a universe $\U$ of possible
solutions. Then $\delta_\cV$ is always monotone and submodular.
\end{proposition}

\reinhard{I have commented out the final sentence of this and the previous section. In principle. deleted }

        \section{Characterizing Volume-based Diversity Functions through the Multi-Attribute Model}
        \label{sec:MultiAttribute}

In~\cite{nehring2002diversitytheory}, Nehring and Puppe proposed a different and novel approach by introducing 
``multi-attribute'' diversity functions. The idea here is to consider {\em attributes} of the elements in a universe $\U$ 
as subsets of $\U$, i.e., each attribute is characterized by the elements that share this attribute. Like Weitzman~\cite{weitzman1992diversity},
the authors drew the motivation for their approach above all from the diversity of species in biodiversity 
and the study of acts that 
ensure high (expected) diversity among them. 
Basic attributes could then be, for instance, ``being a mammal'' or ``living in the ocean'', etc., and each of 
these attributes is then represented by the set of the corresponding animals.

For a finite set $X$, Nehring and
Puppe~\cite{nehring2002diversitytheory} formally define diversity functions as
follows. Let $\lambda$ be a non-negative measure (an additive set
function) on $2^{2^X}$.  For $A \subseteq X$, we write $\lambda_A$
rather than $\lambda(\{A\})$.  Then, for $S \subseteq X$, the
diversity $v_\lambda$ is defined as
\[
v_\lambda(S) \ = \ \lambda(\{A \subseteq X \mid A \cap S \neq \emptyset\}) \ = \ 
\sum_{A \subseteq X \,:\, A \cap S \neq \emptyset} \lambda_A.
\]
Intuitively, this definition considers each subset $A \subseteq X$ as
an attribute (or a ``feature class'') that may contribute to the
diversity of $S$. The weight $\lambda_A$ quantifies the relevance or
distinctiveness of that attribute. A subset $S$ is then considered
diverse if it collectively touches many of these informative subsets
$A$, each with non-negative weight.

We now establish the relationship between this notion and our volume-based approach.
\begin{restatable}{theorem}{thmvolumevsmultiattribute}    
\label{thm:volume:vs:multiattribute}
Let $X$ be a finite set.  If $v_\lambda$ is a multi-attribute
diversity function, then there exists a volume assignment $\cV
=(\cS, \mu, \beta)$ over the universe $\U = X$, such that $v_\lambda
= \delta_\cV$.
Likewise, if  $\cV =(\cS, \mu, \beta)$ is a volume assignment over some {\em finite} universe
$\U$, then there exists a non-negative measure  $\lambda$ on $2^{2^X}$  with  $X = \U$, 
such that $\delta_\cV = v_\lambda$. 
\end{restatable}

\begin{proof}[Proof sketch]
For given 
multi-attribute
diversity function  $v_\lambda$,
defining an equivalent 
volume-based diversity function $v_\cV$ is straightforward. 
More precisely, we set 
$\cV =(\cS, \mu, \beta)$ with 
$\cS =  2^{2^X}$, 
$\beta (x) = \{A \subseteq X\mid x \in A\} $, and 
$\mu(\mathcal{B})  = \sum_{A \in \mathcal{B}} \lambda_A
$ for $\mathcal{B}\in \cS$. 
The other direction is more involved and only works for finite universe $\U$.
In particular, we set 
$\lambda_A = \mu\big(\bigcap_{a\in A} \beta(a) \setminus 
\bigcup_{x \in X \setminus A} \beta(x) \big)$.
\end{proof}

\noindent
This characterization is important for several reasons. First, it confirms
that volume-based diversity functions are at least as expressive as
multi-attribute ones, thereby unifying two frameworks under a common
perspective. Second, the volume-based framework avoids the
computational burden of working directly over the power set of the
power set in the multi-attribute formulation, and
instead operates over a more intuitive geometric or set-based
representation of diversity.
Moreover, by the correspondence with the multi-attribute diversity model, 
our
volume-based diversity functions inherit all favorable properties proved for the
former in \cite{nehring2002diversitytheory}. In particular, the fact
that volume-based diversity functions are monotone and submodular, as
shown in Proposition~\ref{prop:mon-subm}, follows directly from this
equivalence.  
\ifArxiv
Nevertheless, we preferred to give a
self-contained proof of these properties (only using the axioms of the
definition of volume-based diversity functions)
in Appendix~\ref{app:volume}.
\fi

Finally, a fundamental advantage of the volume-based framework is its
suitability for relational data. In this setting, tuples from a
relation can be mapped to measurable regions in a space defined by the
attributes occurring in a tuple or its provenance, allowing the use of volume as a principled
measure of diversity. For example, the balls $\beta(t)$ assigned to
tuples $t$ can reflect their attribute values or provenance sets,
while the measure $\mu$ can reflect weighted or count-based semantics
over these regions. This enables a natural and scalable representation
of diversity across query answers without requiring explicit
enumeration of exponentially many subsets, as is needed in the
multi-attribute approach. In contrast, the latter becomes infeasible
in large relational domains due to its dependence on attribute power
sets. Volume-based diversity is thus more aligned with the semantics
and structure of relational databases.

        \section{Distance-Based versus Volume-Based Diversity Functions}
        \label{sec:DiversityDistance}

As has already been mentioned in Section~\ref{sec:introduction}, a common way of defining diversity of outputs in the database area is by using the distance-based approach. 
In contrast, we have defined the diversity $\delta_\cV$ %
via a volume assignment $\cV = (\cS, \mu, \beta)$ over $\U$. This
raises the question of what are the differences or similarities between the two approaches, and how can we compare them. 

\paragraph{Comparison by properties} A direct way to compare the two approaches is in terms of properties. As we already noticed, volume-based diversity functions are always monotone and submodular. In contrast, almost all distance-based diversity functions are not submodular, and some are not even monotone. 
\ifArxiv
We state this fact here and provide the examples in the appendix.
\fi
\begin{restatable}{proposition}{propdistancenoproperties} \label{prop:distance-no-properties}
	There exists a metric such that its corresponding diversity functions $\deltasum$ and $\deltamin$ are not submodular. Further, $\deltamin$ is 
    not even monotone.
\end{restatable}
Although this fact is direct, it provides evidence that the two approaches differ considerably for $\deltasum$ and $\deltamin$. In the following, we provide further evidence of their differences and similarities.

\paragraph{Volume-based as distance-based} Another way to compare the two approaches is to try to encode volume-based diversity functions by using a distance-based approach. As we will see, in general, this is not possible.
To that end, we first discuss two natural approaches to define a distance function given a volume assignment $\cV=(\cS, \mu, \beta)$.

Specifically, for one, we can define $d^{\triangle}_\cV(a,b)$ as the
measure $\mu$ of the symmetric difference of $\beta(a)$ and
$\beta(b)$, i.e., $d^{\triangle}_\cV(a,b) = \mu \left( (\beta(a) \setminus \beta(b)) \cup
(\beta(b) \setminus \beta(a)) \right)$.  
We observe that any distance measure
$d: \U \times \U \to \mathbb{R}_{\geq 0}$ defined from a volume-based
diversity $\delta_\cV$ as $d:=d^{\triangle}_\cV$ is a {\em pseudo-metric}, i.e., it satisfies \emph{non-negativity} (i.e., $d(a,b) \geq 0$ for all $a,b \in \U$), \emph{symmetry} (i.e., $d(a,b) = d(b,a)$), \emph{identity} ($d(a,a) = 0$ for all $a \in \U$), and the \emph{triangle inequality} (i.e., $d(a,c) \leq d(a,b) + d(b,c)$ for all $a, b, c \in \U$).
If in addition, $d(a,b) = 0$ implies $a = b$ for all $a, b \in \U$,
then $d$ is actually a metric.

A second option (essentially considered in~\cite{nehring2002diversitytheory} in the context of the multi-attribute approach)
is to define the distance function $d^M_{\cV}$ as the marginal $d^M_{\cV}(a,b) := \delta_{\cV}(\{a,b\})-\delta_{\cV}(\{b\})$.
However, in that case, we give up symmetry.
Note that this can be recovered when the diversity of all singletons are the same.
In that case, $d^M_{\cV}$ again becomes a pseudo-metric.

Now, the hope could be that $d^{\triangle}_\cV$ or $d^M_{\cV}$ (or any other pseudo-metric) combined with an appropriate aggregator can recover the expressiveness of $\delta_{\cV}$.
To that end, we denote by $\delta_{\mathit{agg},d}$ a distance-based diversity function defined through an aggregator function $\mathit{agg}$ and a pseudo-metric $d$, namely,  $\delta_{\mathit{agg},d}(S):=\mathit{agg}(d(a,b)_{a,b\in S})$. We say that $\mathit{agg}$ is \emph{monotone} if $\mathit{agg}((d_i)_i)\leq \mathit{agg}((d'_i)_i)$ when $d_i\leq d'_i$ for all $i$.
Further, we say that a volume assignment $\cV = (\cS, \mu, \beta)$ is \emph{oblivious to data values} if for any bijection $f\colon\D\rightarrow \D$ and for any set of tuples $S\subseteq \TUPLES$ we have:
\[
\mu(\bigcup_{R(\ba)\in S}\beta(R(\ba))) = \mu(\bigcup_{R(\ba)\in S}\beta(R(f(\ba)))).%
\]
We also say that pseudo-metric $d$ is \emph{oblivious to data values} if $d(R(\ba),R(\ba'))=d(R(f(\ba)),R(f(\ba')))$.
Essentially, this means that the diversity functions should not depend on the concrete data values that appear as constants in the tuples but instead only on whether constants are equal or not.
Clearly, from the examples presented in Section~\ref{sec:volume}, the volume assignments $\Velem $ and $\Vpos$ are oblivious to data values while $\Velemw, \Vposw, \VQD, \Vr, $ and $\Vg$ are in general not oblivious to data values.
When it comes to metrics, naturally, the Hamming-distance is an example of a metric oblivious to data values while the Euclidean-distance is not.

\begin{restatable}{theorem}{thmnodistance} \label{theo:no-distance}
    There exists a volume assignment $\cV=(\cS,\mu,\beta)$ (e.g., $\Velem$) oblivious to data values over tuples $\TUPLES$ 
     such that there does not exist a monotone aggregator $\mathit{agg}$ and pseudo-metric $d$ over $\TUPLES$ that is oblivious to data values and can distinguish the same sets as $\cV$. In other words, no matter the $\mathit{agg},d$, there are two $k$-subsets $S,S'\subseteq \TUPLES$ such that $\delta_{\cV}(S)\neq \delta_{\cV}(S')$ while $\delta_{\mathit{agg},d}(S) = \delta_{\mathit{agg},d}(S')$.
\end{restatable}

\paragraph{Distance-based as volume-based} We now consider the other direction and see if one can understand the distance-based approach in terms of volumes.
Of course, this is not possible in general as volume-based diversity functions are always monotone and submodular (Proposition~\ref{prop:mon-subm}) while natural distance-based diversity functions are neither (Proposition~\ref{prop:distance-no-properties}). But this leaves open the question if 
more sophisticated distance-based diversity functions like Weitzman's $\deltaW$ can be 
captured by volumes. Below we give a partially positive answer to this question: In general, this is not possible, as we can show that  Weitzman's diversity function $\deltaW$ is, 
in general, not submodular. However, in the most important special case considered in \cite{weitzman1992diversity}, namely if the distance function $d$ 
underlying $\deltaW$ is an ultrametric, then $\deltaW$ is essentially a volume-based diversity function.

For a distance function $d$, The diversity function $\deltaW$ is defined recursively as follows:
\[
\delta_W(S) := \max_{a \in S} \left( \delta_W(S \setminus \{a\}) + d(a, S \setminus \{a\}) \right),
\]
with base case $\delta_W(\{a\}) := 0$. The distance $d(a,S)$ is
defined as $\min_{x \in S} d(a,x)$.

Weitzman’s diversity function is motivated by applications to species
hierarchies. However, one shortcoming is that $\delta_W$ is not
generally submodular:
\begin{restatable}{proposition}{thmweitzmansubmodular}
\label{thm:weitzman:submodular}
Weitzman's diversity measure $\delta_W$ is, in general, not submodular.
\end{restatable}
Another shortcoming of $\delta_W$ is its computational complexity:
even computing $\delta_W(S)$ for a given $S$ is, in general,
intractable \cite{DBLP:journals/corr/abs-2408-01657}. However, if $d$
is an ultrametric (i.e., it satisfies the strong triangle inequality
$d(a,c) \leq \max(\{d(a,b),d(b,c)\})$)
then the computation becomes tractable~\cite{weitzman1992diversity}.
Moreover, in this case, $\delta_W$
becomes essentially volume-based:

\begin{restatable}{theorem}{thmvolumevsweitzman}    
\label{thm:volume:vs:weitzman}
Let Weitzman's diversity measure be defined over a distance function
$d$ that is an ultrametric over some finite set $X$.  Then there
exists a volume assignment $\cV =(\cS, \mu, \beta)$ such that
$\delta_\cV = \delta_W +r$, where $r$ denotes the radius of the
ultrametric (i.e., the max. distance between any two elements in $X$).
\end{restatable}
The above result illustrates a key advantage of our volume-based
framework: it subsumes and generalizes the best-performing cases of
the distance-based approach. In particular, ultrametrics have been
identified as a desirable form of distance for diversity due
to their favorable computational properties \cite{DBLP:journals/corr/abs-2408-01657} and 
their suitability for modeling hierarchical systems~\cite{weitzman1992diversity}.
Note that a hierarchical notion of distance naturally fits relationally structured data
as is illustrated in \cite{%
DBLP:conf/icde/VeeSSBA08,%
DBLP:journals/debu/VeeSA09}, where the distance between two tuples is based on 
the first position at which they differ: tuples with longer common prefixes are 
considered closer. A typical example is a car relation with attributes such as 
`make', `model', `color', and `year'. Under this ultrametric, diversification is done according to the attribute order: first one tries to diversify `make' , then `model', then `color', and finally `year'.

By Theorem~\ref{thm:volume:vs:weitzman}, 
our
framework naturally captures ultrametric diversity functions as a
special case -- up to an additive constant -- through an appropriate volume
assignment. This demonstrates that volume-based diversity not only
provides a broader modeling language for diversity but also inherits
and extends the desirable theoretical guarantees associated with
ultrametric distances. As such, it offers a principled and unified
framework for defining well-behaved diversity measures.

        \section{Query Evaluation Under Volume-Based Diversity Functions}\label{sec:Exact}

In this section, we start our study of CQ evaluation under volume-based diversity functions in data complexity (i.e., the query is fixed). We start by showing that this problem is hard in general for most of the volume assignments $\cV$ presented in Section~\ref{sec:volume}. Despite this negative result, we show that under some reasonable assumptions on $\cV$, we can always find a ($1-1/e$)-approximation of a maximally diverse $k$-subset of the solutions in polynomial time under data complexity.

\paragraph{Hardness of exact computation} Let $\Sigma$ be a schema and $\cV =(\cS, \mu, \beta)$ be a volume assignment over $\TUPLES$. Further, let $Q$ be a CQ over $\Sigma$. We are interested in the following computational~problem:
\begin{center}
		\begin{tabular}{rl}
			\hline\\[-2ex]
			\textbf{Problem:} & $\CQEval[\Sigma, \cV, Q]$\\
			\textbf{Input:} & A database $D$ over $\Sigma$ and $k \geq 1$ \\
			\textbf{Output:} & $\argmax_{S \subseteq \sem{Q}(D)\,:\, |S| = k} \delta_\cV(S)$\\\\[-2.2ex]
			\hline
		\end{tabular}
\end{center}
In other words, given a database $D$ and a number $k \geq 1$, 
we want to compute a $k$-subset $S$ of~$\sem{Q}(D)$ that maximizes the volume diversity $\delta_{\cV}(S)$ over all $k$-subsets.
Note that $\Sigma$ and $Q$ are fixed; namely, we measure the computational resources of the problem in data complexity. 
Furthermore, the volume assignment $\cV$ and, thus, the diversity function $\delta_{\cV}$ are also fixed. We implicitly assume that if $k > |\sem{Q}(D)|$, then we output all the tuples in $\sem{Q}(D)$. In particular, if $\sem{Q}(D) = \emptyset$, then an algorithm for $\CQEval[\Sigma, \cV, Q]$ outputs $\emptyset$. 
By slight abuse of notation, we will formulate intractability results of $\CQEval[\Sigma, \cV, Q]$  in the form of ``$\mathsf{NP}$-hardness''. Strictly speaking, the $\mathsf{NP}$-hardness applies to the decision variant of the problem $\CQEval[\Sigma, \cV, Q]$, i.e., deciding if $\delta_{\cV}(S)$ is above a given threshold $th$ for some $S \subseteq \sem{Q}(D)$ subject to $|S| = k$.

We will always assume that $\cV$ and $\delta_{\cV}$ are fixed in all query evaluation problems studied in this paper (see also Section~\ref{sec:ACQs}). Moreover, for the sake of simplification, in this section we will assume that for any volume assignment $\cV$ and any set $S$ of tuples, computing $\delta_{\cV}(S)$ takes constant time\footnote{We are only making statements on tractability in this section. Section~\ref{sec:ACQs} then focuses on finer analysis {and does not make this assumption}.}. 
Intuitively, one can consider $\delta_{\cV}$ as a black box in the system that can be evaluated efficiently for a set of tuples whose complexity does not considerably affect the query evaluation process. 
Clearly, if we show that $\CQEval[\Sigma, \cV, Q]$ is hard, then it is even harder if the cost of computing $\delta_{\cV}$ is included. The other way around, if we show that $\CQEval[\Sigma, \cV, Q]$ can be evaluated in polynomial time, 
this result will be subjected that $\delta_\cV$ can also be efficiently evaluated (which is typically the case for natural volume assignments $\cV$).

Unfortunately, similar to previous work on query evaluation under diversity functions, we can show that $\CQEval$ is $\NP$-hard for most of the volume assignments $\cV$ presented in Section~\ref{sec:volume}.

\begin{restatable}{theorem}{theoCQhardness}    
\label{theo:CQ-hardness}
	The problem $\CQEval[\Sigma, \cV, Q]$ is $\NP$-hard if $\cV \in \{\Velem, \Vpos,\Velemw, \Vposw, \VQD\}$.
\end{restatable}

Given that for simple volume assignments like $\Velem$ and $\Vpos$, the query evaluation problem is hard, we move in the rest of this section to provide good approximations to $\CQEval[\Sigma, \cV, Q]$.

\paragraph{Approximation of optimal solutions} Recall that $\Sigma$ is a schema, $Q$ is a CQ over $\Sigma$, and $\cV$ is a volume assignment over $\TUPLES$. We say that $S^* \subseteq \sem{Q}(D)$ with $|S^*| = k$ is an \emph{$(1-\epsilon)$-approximation} of $\CQEval[\Sigma, \cV, Q]$ on a database $D$ and a number $k \geq 1$ if, and only if:
\[
\delta_\cV(S^*) \ \geq \ (1-\epsilon) \cdot \max_{S \subseteq \sem{Q}(D)\,:\, |S| = k} \delta_\cV(S)
\]
In other words, the diversity of $S^*$ with respect to $\delta_\cV$ is not worse than $(1-\epsilon)$ times 
the diversity of the best solution, where the smaller $\epsilon \geq 0$, the better the approximation. 

Since $\CQEval[\Sigma, \cV, Q]$ is $\NP$-hard, we strive to find an $(1-\epsilon)$-approximation for some $\epsilon \geq 0$. 
Given that $\delta_\cV$ is monotone and submodular by Proposition \ref{prop:mon-subm}, we can take advantage of the algorithmic theory of submodular set functions to find the following approximation~\cite{nemhauser1978analysis}. 

\begin{algorithm}[t]
	\DontPrintSemicolon
	\KwIn{A database $D$ and a value $k \geq 1$.}
	\KwOut{A $k$-diversity set $S \subseteq \sem{Q}(D)$ with respect to $\delta_{\cV}$.}
	$S \gets \emptyset$ \\
	\For{$i=1$ \KwTo $k$}{
		$t^* \gets \argmax_{t \in \sem{Q}(D)} \delta_{\cV}(S \cup \{t\})$ \label{line:greedy}\\
		$S \gets S \cup \{t^*\}$ \\
	}
	\Return{$S$}
	\caption{Greedy algorithm for finding a $(1-\nicefrac{1}{e})$-approximation of the problem $\CQEval[\Sigma, \cV, Q]$ for a schema $\Sigma$, a volume assignment $\cV$, and a CQ $Q$ over~$\Sigma$.}
	\label{alg:greedy}
\end{algorithm}

\begin{restatable}{theorem}{theoCQapproximation}    
\label{theo:CQ-approximation}
	One can compute an $(1-\nicefrac{1}{e})$-approximation of $\CQEval[\Sigma, \cV, Q]$ for every database $D$ and $k \geq 1$ in polynomial time in $|D|$, where $e$ is the Euler number. %
\end{restatable}

\begin{proof}%
	In~\cite{nemhauser1978analysis}, Nemhauser, Wolsey, and Fisher showed that for every monotone submodular set function $f:\finite(\U) \rightarrow \bbR$ and $k\geq 1$ one can compute in polynomial time a $k$-subset~$A$ of $\U$ such that $f(A) \geq (1-\nicefrac{1}{e}) \cdot \max_{B \subseteq \U: |B| = k} f(B)$.
	Since  $\delta_\cV$ is submodular and monotone and $Q$ is fixed, one can compute the set $\sem{Q}(D)$ in polynomial time over $D$ and then apply the result in~\cite{nemhauser1978analysis} to retrieve a  $(1-\nicefrac{1}{e})$-approximation of $\delta_\cV$ over  $\sem{Q}(D)$ restricted to subsets of~size~$k$. In Algorithm~\ref{alg:greedy}, we depict this procedure for $\delta_{\cV}$ which follows a greedy strategy: starting from $S = \emptyset$; in every iteration it finds a tuple $t \in \sem{Q}(D)$ that maximizes the \emph{marginal diversity} of $\delta_\cV$, namely, $\delta_\cV(S \cup \{t\}) \setminus \delta_\cV(S)$. After the $k$-th iteration,  it outputs $S$. By~\cite{nemhauser1978analysis}, this procedure achieves a $(1-\nicefrac{1}{e})$-approximation of $\CQEval[\Sigma, \cV, Q]$ for every database $D$ and $k \geq 1$, and runs in polynomial time.
\end{proof}

The previous result is indeed a direct consequence of Nemhauser et al. techniques on the maximization of submodular set functions. Nevertheless, one must compare the approximation ratio obtained for volume-based diversity functions with that of the best approximation found for distance-based analogs. Recently, approximation algorithms were proposed in~\cite{DBLP:journals/pacmmod/AgarwalEHSY24} for CQ evaluation under distance-based diversity functions. For $\deltamin$ 
under the Hamming or Euclidean metrics, the best approximation ratio is $(1 - (\nicefrac{1}{2} + \epsilon))$, and the running time of the algorithms depends on $\epsilon$. For $\deltasum$, the best approximation ratio is $(1 - \nicefrac{2}{k})$ for Hamming distance (for Euclidean distance it is $(1-\nicefrac{1}{2})$) but the running time depends on $k^3$. 
Instead, the approximation ratio for volume-based diversity functions is $(1-\nicefrac{1}{e})$ 
and works \emph{for every volume assignment} that can be computed in polynomial time (in particular, for most of the examples presented in Section~\ref{sec:volume}). Furthermore, Algorithm~\ref{alg:greedy} can be easily incorporated into the current query evaluation strategy of any database management system by finding all tuples in $\sem{Q}(D)$ and then applying Algorithm~\ref{alg:greedy}. 

We want to end this section by showing that, in general,  $(1-\nicefrac{1}{e})$-approximation is the best one 
can get  for volume-based diversity functions.

\begin{restatable}{theorem}{theoCQapproximationhardness}    
\label{theo:CQ-approximation-hardness}
	There exists a schema $\Sigma$, a volume assignment $\cV$, and a CQ $Q$ such that a $(1-\nicefrac{1}{e})$-approximation of $\CQEval[\Sigma, \cV, Q]$ is the best that one can get in polynomial time data complexity, unless~$\PTIME = \NP$. 
\end{restatable}

\begin{proof}[Proof sketch]
	The proof is by encoding the maximum coverage problem into a volume assignment $\cV$. 
	It is well-known that the maximum coverage problem is hard to approximate beyond $(1-\nicefrac{1}{e})$-approximation ratio, unless $\PTIME = \NP$~\cite{DBLP:journals/jacm/Feige98}.
\end{proof}

	\section{Approximating Volume-based Diverse Answers Under Combined Complexity}\label{sec:ACQs}

In the following, we aim to lift the results of Section~\ref{sec:Exact} to the combined complexity case and provide a finer analysis.
{Note, in this section, we include the time required to compute $\delta_{\cV}(S)$ in our analysis.} We start by stating the main problem and recalling some standard notation for efficient CQ evaluation. Then, we present the main approach for efficient CQ  evaluation under volume-based diversity functions and apply it to some specific volume assignments. We conclude by demonstrating how to generalize the technique by connecting it to the ranked enumeration problem of CQ evaluation. 

\paragraph{Problem statement and main definitions}
In this section, we aim to solve the following problem:
\begin{center}
		\begin{tabular}{rl}
			\hline\\[-2ex]
			\textbf{Problem:} & $\CQEval[\Sigma, \cV]$\\
			\textbf{Input:} & A database $D$ and a CQ $Q$ over $\Sigma$, and $k \geq 1$ \\
			\textbf{Output:} & $\argmax_{S \subseteq \sem{Q}(D)\,:\, |S| = k} \delta_\cV(S)$\\\\[-2.2ex]
			\hline
		\end{tabular}
\end{center}
where $\Sigma$ and $\cV$ are a fixed schema and a fixed volume assignment.
Contrary to Section~\ref{sec:Exact}, we cannot afford 
to find a $(1-\nicefrac{1}{e})$-approximation by first computing $\sem{Q}(D)$ (whose size is $O(|D|^{|Q|})$)
and then applying Algorithm~\ref{alg:greedy}. 
In other words, the set $\sem{Q}(D)$ is compactly represented by $(Q, D)$, and the challenge is to find the most diverse $k$-subset or an approximation without computing $\sem{Q}(D)$.

Recall that even determining the existence of answers to CQs is $\NP$-hard in 
combined complexity~\cite{DBLP:conf/stoc/ChandraM77}.
Thus, we will restrict ourselves to CQs with bounded \textit{fractional hypertree width} ($\fhw$) \cite{GroheM14}. 
To that end, we briefly recall the notions of tree decompositions and $\fhw$.

    Let $Q(\bx) \leftarrow R_1(\bx_1), \ldots, R_{m}(\bx_{m})$ be a CQ using variables in $\Var$. 
    For the sake of simplification, in the sequel, we assume that every sequence $\bx_i$ does not repeat variables and, thus, 
    by slight abuse of notation, we may treat $\bx_i$ as a set (otherwise, one can remove duplicate variables by rewriting $Q$ and preprocessing $D$ in linear time w.r.t.\ $|D|$).
    A \emph{tree decomposition} of $Q$ is a tuple $(T,\chi)$ where $T=(V(T),E(T))$ is a rooted tree and $\chi\colon V(T) \mapsto 2^{\Var}$ assigns to each $v\in V(T)$ a subset $\chi(v)\subseteq \Var$ called a \emph{bag}.
    Additionally, the following properties have to be satisfied:
    \begin{enumerate}
        \item for every variable $x\in \Var$, the set $\{v \in V(T) \mid x\in \chi(v)\}$ induces a connected subtree of $T$; and
        \item for every relation $R_i(\bx_i)$, there exists $v\in V(T)$ that contains all of $\bx_i$ in its bag $\chi(v)$.
    \end{enumerate}
    The fractional hypertree width of a tree decomposition $(T,\chi)$  is $\max_{v \in V(T)} \rho^*(\chi(v))$ where $\rho^*(\chi(v))$ is the minimum fractional edge cover of the hypergraph induced by $\chi(v)$ over $Q(\bar{x})$. 
    The \emph{fractional hypertree width $\fhw(Q)$ of $Q$} is the minimum fractional hypertree width among all tree decompositions of $Q$. 
   	Finally, a conjunctive query is called an \emph{acyclic CQ} (ACQ) iff $\fhw(Q)=1$.

\paragraph{Approximation through maximizing the marginal diversity} 
Motivated by Theorem~\ref{theo:CQ-approximation} and Algorithm~\ref{alg:greedy}, 
a reasonable strategy to find an approximation for $\CQEval[\Sigma, \cV]$ is to compute the next tuple $t$ that maximizes the marginal diversity of $\delta_\cV(S)$. 
In other words, we have to consider the problem of computing greedily the next best solution (see line 3 in Algorithm~\ref{alg:greedy}):
\begin{center}
		\begin{tabular}{rl}
			\hline\\[-2ex]
			\textbf{Problem:} & $\CQNext[\Sigma, \cV]$\\
			\textbf{Input:} & A database $D$ and a CQ $Q$ over $\Sigma$, and a subset $S\subseteq \sem{Q}(D)$ \\
			\textbf{Output:} & $\argmax_{t\in \sem{Q}(D)} \delta_{\cV}(S\cup \{t\})$\\\\[-2.2ex]
			\hline
		\end{tabular}
\end{center}
Similar to $\CQEval[\Sigma, \cV]$, the main challenge is to compute $t$ from $D$, $Q$, and $S$, without necessarily computing $\sem{Q}(D)$. Naturally, if we can solve $\CQNext[\Sigma, \cV]$ efficiently, then we can apply Algorithm~\ref{alg:greedy} by calling $\CQNext[\Sigma, \cV]$ in line~\ref{line:greedy} and solve $\CQEval[\Sigma, \cV]$. In other words, we get the following result.

\begin{restatable}{theorem}{thmimpapprox}    
	\label{thm:impapprox}
	If $\CQNext[\Sigma, \cV]$ can be solved in time $O(f)$ for some function $f$, then the problem $\CQEval[\Sigma, \cV]$ can be $(1-\nicefrac{1}{e})$-approximated in time $O(k\cdot f)$.
\end{restatable}

The converse of Theorem \ref{thm:impapprox} does not necessarily hold. {In particular, we do not know whether hardness of $\CQNext[\Sigma, \cV]$ implies that $\CQEval[\Sigma, \cV]$ cannot be approximated (see Section~\ref{sec:conclusions} for further ideas).}
However, at least, if $\CQNext[\Sigma, \cV]$ is $\NP$-hard for the singleton case (fixing $S=\emptyset$), also the problem (exact version) $\CQEval[\Sigma, \cV]$ must be $\NP$-hard (for $k=1$).

We now revisit the volume assignments from Section~\ref{sec:volume}
and separate the hard and easy cases for solving $\CQNext[\Sigma, \cV]$.
We start with the hard cases which, thus, do not translate to approximability results of $\CQEval[\Sigma, \cV]$:

\begin{restatable}{theorem}{theoHardnessCompbinedComplexity}    
\label{theo:hardnessCombinedComplexity}
Unless $\mathsf{P}=\mathsf{NP}$,
	the problem $\CQNext[\Sigma, \cV]$ cannot be solved in polynomial time for  $\cV \in \{\Velem, \Velemw, \VQD\}$,
    even if we only allow ACQs and subsets $S=\emptyset$.
\end{restatable}

\begin{proof}[Proof sketch]
We  illustrate the basic idea by proving $\NP$-hardness of the apparently simplest case $\cV = \Velem$.
The proof is by reduction from (the directed version of) Hamiltonian path:
Given an instance $G = (V(G), E(G))$ of Hamiltonian path, we define an instance $(D,Q,S)$ of $\CQNext[\Sigma, \cV]$
as follows: database $D$ consists of a single binary relation $E$ storing the edges of $G$, we set $ S= \emptyset$,
and, for $n = |V(G)|$, we define the ACQ $Q$ as follows:
        \[
        Q(x_1,\dots,x_{n}) \ \leftarrow \  E(x_1,x_2),\dots,E(x_{n-1},x_{n}).
        \]
        Notice that a solution $Q(h(\overline{x}))\in \sem{Q}(D)$ corresponds to a walk in $G$ and $\delta_{\Velem}$ applied to singletons (i.e., $\delta_{\Velem}(\{Q(h(\overline{x}))\})$) counts the number of distinct vertices used in the corresponding walk. Hence,  $G$ is a positive instance of Hamiltonian path if, and only if, 
        the solution to this instance of $\CQNext[\Sigma, \Velem]$ yields an answer $Q(h(\overline{x}))$ with 
        $\delta_{\Velem}(\{Q(h(\overline{x}))\}) = n$.
\end{proof}

Next, we show that even seemingly simple changes in the diversity function can affect the tractability of $\CQNext[\Sigma, \cV]$ and, thus, 
naturally lead to the approximability of $\CQEval[\Sigma, \cV]$ due to Theorem~\ref{thm:impapprox}.

\begin{restatable}{theorem}{theoTractableCompbinedComplexity}    
\label{theo:tractableCombinedComplexity}
Restricted to ACQs, the problem $\CQNext[\Sigma, \cV]$ can be solved in time $O(|Q|\cdot |D|)$ for  
$\cV \in \{\Vpos, \Vposw\}$ when only allowing ACQs.
Hence, in this case, 
$\CQEval[\Sigma, \Velem]$ can be $(1-\nicefrac{1}{e})$-approximated in time $O(k\cdot|Q|\cdot |D|)$.
\end{restatable}

\begin{proof}[Proof sketch]
We explain why  $\CQNext[\Sigma, \Vpos]$ is tractable for ACQs $Q(\bx)$.
To that end, let $D$ be a database, %
and $h_1,\dots h_k$ homomorphisms from $Q$ to $D$, i.e., $S=\{Q(h_1(\bx)),\dots,$ $ Q(h_k(\bx))\}\subseteq \sem{Q}(D)$.%
Consider the \textit{marginal} diversity for a new solution $Q(h(\bx))\in \sem{Q}(D)$:
\[
\delta_{\Vpos}(S\cup \{Q(h(\bx))\})-\delta_{\Vpos}(S) \ = \  \sum_{x\in \bx}\alpha_{x}, \quad \text{with }\alpha_{x}=\begin{cases}
1 & \text{if } \forall i \colon h(x)\neq h_i(x),\\
0 & \text{if } \exists i \colon h(x) = h_i(x).\\
\end{cases}
\]
That is, we count the number of \textit{new} values.
We can cast this then as a sum-product query over the tropical semi-ring $\mathbb{R}_{\max}:=(\mathbb{R}\cup\{\infty\}, +,\max)$.
Doing so shows that we can find the element that maximizes the marginal diversity in linear time.
To do so, for every $x\in \bx$ let us choose a \textit{covering} relation $R^x:=R_i(\bx_i)$ used in $Q$ where $x\in \bx_i$.
Then, we can define the $\mathbb{R}_{\max}$-relations $R^*_1,\dots,R^*_m$.
That is, for tuple $R_j(\ba)$ in the database, we add tuples $R^*_j(\ba)$ to the database and annotate it with the number of \textit{new} values $\ba$ adds at positions $x$ such that $R_j$ covers $x$.
Then, as every variable $x\in \bx$ is covered by exactly one relation, we have: 
\begin{align}
    \delta_{\Vpos}(S\cup \{Q(h(\bx))\})-\delta_{\Vpos}(S) = \sum_{i} R_i^*(h(\bx_i)) \label{eq:spq}
\end{align}
for homomorphism $h$ such that $Q(h(\bx))\in \sem{Q}(D)$.
Then, due to results on sum-product queries~\cite{DBLP:conf/pods/KhamisNR16,DBLP:journals/jcss/PichlerS13}
we can find a $Q(h(\bx))\in \sem{Q}(D)$ maximizing Equation \eqref{eq:spq} in time $O(|Q|\cdot |D|)$ as this is then a scalar sum-product query.
This solves $\CQNext[\Sigma, \Velem]$ for ACQs.
Then, due to Theorem \ref{thm:impapprox}, we can compute a $(1-\nicefrac{1}{e})$-approximation of $\CQEval[\Sigma, \Velem]$ in time $O(k\cdot|Q|\cdot |D|)$.
\end{proof}

\paragraph{Diverse answers to CQs via ranked enumeration}
Towards a more general criterion to ensure tractability of $\CQNext[\Sigma, \cV]$,
we consider this problem as a top-$k$ ranked enumeration problem, where the marginal diversity is the value by which we order the output and where we ask for the top-1 answer (we can ignore the additive constant $\delta_{\cV}(S)$).
Actually, top-$k$ ranked enumeration has received considerable attention from the database community in the last years
(see e.g., \cite{deep2025ranked,DBLP:conf/sigmod/IlyasSAVE04,%
	DBLP:conf/sigmod/LiCIS05,%
	DBLP:conf/vldb/LiSCI05}), where we consider \cite{deep2025ranked} as the most general and most naturally extendable to our setting.
	
\cristian{I moved the above paragraph here (after the subsection). I think that it better fits here as an introduction to it.}

We briefly recall the setting and main result of \cite{deep2025ranked} and then build on them.
There, rank functions $\texttt{rank}$ assign values $\texttt{rank}(Q(h(\bx)))\in \mathbb{R}$ to 
solutions of CQs $Q(h(\bx))\in \sem{Q}(D)$ and the goal is to enumerate $Q(h(\bx))\in \sem{Q}(D)$ 
in the order induced by $\texttt{rank}$, i.e., $Q(h(\bx))$ should be output before $Q(h'(\bx))$ if $\texttt{rank}(Q(h(\bx))) > \texttt{rank}(Q(h(\bx)))$.
Informally speaking, the main result of \cite{deep2025ranked} is that,
with the help of a tree decomposition $(T,\chi)$ of the full CQ $Q$, enumeration is efficiently possible 
if $\texttt{rank}$ is \textit{compatible} with $(T,\chi)$.

Given a volume assignment $\cV = (\cS, \mu, \beta)$, we would like to apply the results of \cite{deep2025ranked} 
to the functions $\texttt{rank}_{\cV, S} :=\delta_{\cV}(S \cup \{\cdot\})$.
Thus, naively, we would have to verify that $\texttt{rank}_{\cV, S}$ is compatible with $(T,\chi)$ for every $S\subseteq \sem{Q}(D)$.
Inspired by their use of compatibility, in the remainder of this section, we develop a notion of compatibility (with a tree decomposition $(T,\chi)$) of the ball function $\beta$.
This will be a sufficient condition, such that $\texttt{rank}_{\cV, S}$ is compatible with $(T,\chi)$ for every $S\subseteq \sem{Q}(D)$.
To that end, we start as in \cite{deep2025ranked} by defining what it means (in our case for $\beta$) to be $\overline{y}$-decomposable.

\begin{definition}
\label{def:decomp}
Let $\cV = (\cS, \mu, \beta)$ be a volume assignment, $R(\bx)$ be an atom over $\Sigma$ with variables $\bx$, and $\overline{y}\subseteq \overline{x}$.
We say that $\beta$ is {\em $\overline{y}$-decomposable} (w.r.t.\ $R$) if for every pair of homomorphisms $h, h'$ over $\overline{y}$ and homomorphisms $g,g'$ over $\overline{x} \setminus \overline{y}$ we have:
\begin{align}
\beta(R((h\cup g)(\overline{x}))) \setminus \beta(R((h'\cup g)(\overline{x}))) = \beta(R((h\cup g')(\overline{x}))) \setminus \beta(R((h'\cup g')(\overline{x}))). \label{eq:decom}
\end{align}
\end{definition}

The intuition of $\overline{y}$-decompositions is the following: Whatever a partial homomorphism $h$ on $\overline{y}$ 
contributes to the volume compared with another partial homomorphism $h'$ should not depend on how $h$ and $h'$ 
are completed (i.e., either by $g$ or $g'$).
Let us denote the set in Equation \eqref{eq:decom} as $\beta(h,h')$.

For a set $S$ of $R$-tuples, let us now consider the function $\texttt{rank}_{\cV,S}$ defined for $R$-tuples.
Then, to compare the function value of $\texttt{rank}_{\cV,S}$ on two homomorphisms $\hat{h}$ and $\hat{h}$ that agree outside of $\by$, 
it suffices to compare $\mu(\beta(h,h')\setminus \bigcup_{s\in S}\beta(s))$
with $\mu(\beta(h',h)\setminus \bigcup_{s\in S}\beta(s))$.
Consequently, the function $\texttt{rank}_{\cV,S}$ is $\overline{y}$-decomposable in the sense of \cite{deep2025ranked} for every set $S$ 
\ifArxiv
(we explain this in more detail in Appendix~\ref{app:acq}).
\fi

Thus, to extend the main result of \cite{deep2025ranked} to our setting, we can extend our notion of decomposability to compatibility w.r.t.\ a tree decomposition analogously to how it is done there.
We note that while Definition \ref{def:decomp} significantly differs from the counterpart in \cite{deep2025ranked},
extending it to compatibility is rather immediate.
Thus, we only give the following definitions for the sake of completeness.

Let $\cV = (\cS, \mu, \beta)$ be a volume assignment, let $R(\bx)$ be an atom over $\Sigma$ with variables $\bx$, and let $\overline{y},\overline{z}\subseteq \overline{x}$ be such that $\overline{y}\cap \overline{z} = \emptyset$.
Further, let $R_{\overline{x}\setminus\overline{z}}\not \in \Sigma$ be a new relation symbol of arity $|\overline{x}\setminus\overline{z}|$.
We say that $\beta$ is {\em $\overline{y}$-decomposable conditioned on $\overline{z}$} (w.r.t.\ $R$)
if for every homomorphism $f$ over $\overline{z}$, the ball function extended to $R_{\overline{x}\setminus\overline{z}}$-tuples via
$\beta(R_{\overline{x}\setminus \overline{z}}(\hat{h}(\overline{x}\setminus \overline{z}))) := \beta(R((\hat{h}\cup f)(\overline{x})))$ for homomorphism $\hat{h}$ over $\overline{x}\setminus\overline{z}$ is $\overline{y}$-decomposable w.r.t\ $R_{\overline{x}\setminus\overline{z}}$.

Let $(T,\chi)$ be a rooted tree decomposition of a full CQ $Q(\bx)$. 
For $t\in V(T)$ we denote with $\chi(T_t)$ the union of the bags in the subtree rooted in $t$.
Further, with $\texttt{key}(t)$ we denote the variables $\chi(t)\cap \chi(p)$ where $p$ is the parent of $t$ and $\texttt{key}(r)=\emptyset$ for the root $r$ of $T$.
We say that $\beta$ is {\em compatible with $(T,\chi)$} if for every node $t$ it is $(\chi(T_{t}) \setminus \texttt{key}(t))$-decomposable conditioned on $\texttt{key}(t)$ w.r.t.\ $Q$.

As explained before, since $\overline{y}$-decomposability in our sense can be reduced to $\overline{y}$-decomposability for every set $S$ in the sense of \cite{{deep2025ranked}}, we get the following by combining it with Theorem \ref{thm:impapprox}.

\begin{restatable}{theorem}{thmappeval}    
\label{thm:appeval}
    Let $\cV=(\cS, \mu, \beta)$ be a volume assignment over $\TUPLES$ such that $\beta$ is compatible with a rooted tree decomposition $(T,\chi)$ of the full CQ $Q(\bx)$.
    Then, $\CQEval[\Sigma, \cV]$ can be $(1-\nicefrac{1}{e})$-approximated in time $O(|Q|\cdot|D|^{fhw(T,\chi)}\cdot k \cdot T_{\cV})$ where $T_{\cV}$ is the time
    to compute marginals of $\delta_{\cV}$ for fixed sets.
\end{restatable}

To showcase Theorem \ref{thm:appeval}, we revisit the volume assignment $\VQD$ from
Example~\ref{ex:VQD}. 

\begin{restatable}{theorem}{thmprovenance}    
\label{thm:provenance}
    Let $Q(\bx)$ be a CQ such that every atom $R_i(\bx_i)$ of $Q$ uses a unique relation name and let $(T,\chi)$ be a tree decomposition of $Q$ such that there is a subtree $T_{\bx}$ of $T$ containing the root of $T$ and where $\bx=\bigcup_{v\in V(T_{\bx})}\chi(v)$. That is, the CQ is self-join-free and the tree decomposition is free-connex \cite{DBLP:conf/csl/BaganDG07}. Then
    $\CQEval[\Sigma, \VQD]$ can be $(1-\nicefrac{1}{e})$-approximated in time $O(|Q|\cdot|D|^{fhw(T,\chi)+1}\cdot k)$.
\end{restatable}

We juxtapose it with Theorem \ref{theo:hardnessCombinedComplexity}: %
In Theorem \ref{theo:hardnessCombinedComplexity} we say that $\CQNext[\Sigma, \VQD]$
is intractable even for ACQs while we now state that computing a $(1-\nicefrac{1}{e})$-approximation of $\CQEval[\Sigma, \VQD]$
is tractable for CQs when $fhw(T,\chi)$ is small. The crucial restriction in Theorem~\ref{thm:provenance} is 
{\em self-join-freeness}, which is in effect similar to keeping positions apart as $\Vpos$ does compared to $\Velem$.

\begin{proof}[Proof Sketch of Theorem~\ref{thm:provenance}]
    Theorem \ref{thm:appeval} cannot directly be applied since $Q$ is not necessarily a full CQ.
    To that end, let us consider the full CQ $Q^{\bx}(\bx)$ defined as the subquery of $Q$ where all body relations are projected onto $\bx$.
    Then, $(T_{\bx},\chi|_{V(T_{\bx})})$ is a tree decomposition of $Q^{\bx}$ and $fhw(T_{\bx},\chi|_{V(T_{\bx})})\leq fhw(T,\chi)$.
    Now, we extend $\betaQD$ to $Q^{\bx}$-tuples via $\betaQD(Q^{\bx}(h(\bx))):=\betaQD(Q(h(\bx)))$.
    Defined as such, $\betaQD$ is compatible with $(T_{\bx},\chi|_{V(T_{\bx})})$ w.r.t.\ $Q^{\bx}$ as $Q$ and, hence, also $Q^{\bx}$ are self-join-free.

    Then, to compute $\texttt{rank}_{\VQD, S}$, we have to keep track of the which-provenance \cite{DBLP:journals/vldb/CuiW03} for each of the tuples in the bags of $v\in V(T_{\bx})$ for what happens ``outside'' of $V(T_{\bx})$.
    Thus, essentially, for each $v\in V(T_{\bx})$, we have to look at its children in $T\setminus T_{\bx}$, i.e., $C:=child(v)\setminus V(T_{\bx})$ and consider the sub-query $Q^v(\chi(v))$ that uses the variables $\chi(v)$, and the ones that appear in $C$ and their descendants.
    Computing the provenance of these queries requires time $|D|^{fhw(T,\chi)+1}$ (where the $+1$ is to account for the semi-ring 
    operations)~\cite{DBLP:conf/pods/KhamisNR16, DBLP:journals/jcss/PichlerS13}.
    However, then, to compute marginals of $\delta_{\cV}$ (essentially $\delta_{\cV}(S\cup \{s\})$), it suffices to add together the provenance of every tuple $t\in S \cup \{s\}$.
    The provenance of a tuple $t$ can be computed by looking-up and adding together the provenance of $t$ projected to $\chi(v)$ in $Q^v$.
    Thus, as $S$ can be considered fixed, this takes $O(|D|)$ time.
\end{proof}

In particular, this means that for self-join-free, free-connex, acyclic conjunctive queries, the problem $\CQEval[\Sigma, \VQD]$ can be $(1-\nicefrac{1}{e})$-approximated in quadratic time (for constant $Q,k$).

        \section{Conclusions}\label{sec:conclusions}

In this work, we have introduced the volume-based framework for diversity measures $\delta_\cV$, providing several examples of them in relational databases, and we have studied their properties. Above all, given the intractability of query answering under diversity, 
we have shown an approximation algorithm that runs in polynomial time data complexity, and we have identified criteria for extending the tractability of the approximation to combined complexity.
Arguably, all these results provide substantial evidence that volume-based diversity forms an alternative approach to distance-based diversity, which requires further consideration in both the theory and practice of database management systems.

For future work, we propose to take a closer look into the relationship 
between our framework of volume-based diversity measures and the distance-based approach. 
In Section~\ref{sec:DiversityDistance}, we have shown that Weitzman's (distance-based) diversity function $\deltaW$  
essentially becomes a volume-based diversity function if the underlying distance function is
an ultrametric. In \cite{DBLP:journals/corr/abs-2408-01657}, general criteria were presented that make the problem of computing the 
exact solution of $\CQEval$ tractable if the distance underlying a diversity measure is an 
ultrametric. 
It would be interesting to explore restrictions under which this can be lifted to volume-based diversity functions. 

Another interesting open problem is to find other strategies for approximating $\CQEval$ in combined complexity. According to Theorem~\ref{theo:hardnessCombinedComplexity}, $\CQNext$ cannot be solved in polynomial time (under complexity assumptions) for many natural volume-based diversity functions. Nevertheless, even in the cases where $\CQNext$ is NP-hard, one might still be able to get a reasonable approximation algorithm. 
This approximation algorithm $\CQNext$, combined with Algorithm~\ref{alg:greedy}, could then lead to 
an approximation of $\CQEval$. 

        \section*{Acknowledgements}
        
        The work of Merkl and Pichler was supported  by the Vienna Science and Technology Fund (WWTF) [10.47379/ICT2201, 10.47379/VRG18013, 10.47379/NXT22018]. The work of Arenas and Riveros was supported by ANID – Millennium Science Initiative Program – Code ICN17\_002. Riveros was also supported by ANID Fondecyt Regular project 1230935.

	\bibliographystyle{abbrv}
	\bibliography{extras/biblio}
	
	\newpage
	\appendix

	\section{Additional Details for Section \ref{sec:volume}}
    \label{app:volume}
In Proposition~\ref{prop:mon-subm}, we have claimed that, 
for arbitrary volume assignment  $\cV =(\cS, \mu, \beta)$ over an arbitrary universe $\U$,
the distance function  $\delta_\cV$ is monotone and submodular.
Monotonicity follows trivially from the definition of $\delta_\cV(S)$ 
by applying a monotone measure $\mu$ to the union of balls $\beta(a)$ with $a \in S$. 
That is, the bigger $S$, the bigger the union of balls, the bigger $\delta_\cV(S)$.
In the sequel, we thus concentrate on submodularity.

\begin{proof}[Proof of submodularity.]
Consider an arbitrary volume assignment $\cV =(\cS, \mu, \beta)$ over 
some universe $\U$. 
Let $S_1, S_2 \subseteq \U$ and $a \in \U$ such that $S_1 \subseteq S_2$.
We have to show that the following condition holds:
\begin{align}\label{eq-sub-mod}
\delta_\cV(S_1 \cup \{a\}) - \delta_\cV(S_1) \ \geq \
\delta_\cV(S_2 \cup \{a\}) - \delta_\cV(S_2).
\end{align}
W.l.o.g., let us assume that $a \not\in S_2$ since, otherwise, 
$S_2 \cup \{a\} = S_2$ and the inequality 
 \eqref{eq-sub-mod}  holds trivially. 

Given that $\delta_\cV$ is a volume diversity function over $\cV$, we have that:
\begin{align*}
    \delta_\cV(S_1 \cup \{a\}) \ &= \ \mu(\bigcup_{x \in S_1 \cup \{a\}} \beta(x))\\
    &= \ \mu(\bigcup_{x \in S_1} \beta(x) \cup \beta(a))\\
    &= \ \mu(\bigcup_{x \in S_1} \beta(x) \cup (\beta(a) \smallsetminus \bigcup_{x \in S_1} \beta(x)))\\
    &= \ \mu(\bigcup_{x \in S_1} \beta(x)) + \mu(\beta(a) \smallsetminus \bigcup_{x \in S_1} \beta(x)).
\end{align*}
Then, given that $\delta_\cV(S_1) = \mu(\bigcup_{x \in S_1} \beta(x))$, we conclude that:
\begin{align}\label{eq-X-sm}
    \delta_\cV(S_1 \cup \{a\}) - \delta_\cV(S_1)\ = \ \mu(\beta(a) \smallsetminus \bigcup_{x \in S_1} \beta(x)).
\end{align}
In the same way, we conclude that:
\begin{align}\label{eq-Y-sm}
    \delta_\cV(S_2 \cup \{a\}) - \delta_\cV(S_2)\ = \ \mu(\beta(a) \smallsetminus \bigcup_{y \in S_2} \beta(y)).
\end{align}
We are assuming $S_1 \subseteq S_2$. Hence, we have $\bigcup_{x \in S_1} \beta(x) \subseteq \bigcup_{y \in S_2} \beta(y)$, from which we conclude that:
\begin{align*}
\beta(a) \smallsetminus \bigcup_{y \in S_2} \beta(y)
\ \subseteq \ 
\beta(a) \smallsetminus \bigcup_{x \in S_1} \beta(x).
\end{align*}
Given that $\mu$ is a measure, we conclude that:
\begin{align*}
\mu(\beta(a) \smallsetminus \bigcup_{y \in S_2} \beta(y))
\ \leq \ 
\mu(\beta(a) \smallsetminus \bigcup_{x \in S_1} \beta(x)).
\end{align*}
Then, combining this property with \eqref{eq-X-sm} and \eqref{eq-Y-sm}, we conclude that \eqref{eq-sub-mod} holds.
\end{proof}

	\section{Additional Details for Section \ref{sec:MultiAttribute}}
    \label{app:sec:MultiAttribute}
Recall Theorem \ref{thm:volume:vs:multiattribute} from Section~\ref{sec:MultiAttribute}:

\thmvolumevsmultiattribute*

\begin{proof}
First, consider a diversity function $v_\lambda$ over some finite set
$X$ according to the multi-attribute approach
of \cite{nehring2002diversitytheory}, i.e., $\lambda$ is a
non-negative measure on $2^{2^X}$.
We define a volume-based diversity function $\delta_{\cV}$ with the
following volume assignment $\cV =(\cS, \mu, \beta)$:
\begin{itemize}
    \item $   \cS =  2^{2^X}$; \quad\quad  i.e., sets of subsets of $X$. 
    \item $   \beta \colon X \rightarrow \mathbb{R}_{\geq 0}$ 
     \quad with $\beta (x) = \{A \mid x \in A\} $
\item for $\mathcal{B} \subseteq \cS$, we set  
$\mu(\mathcal{B})  = \sum_{A \in \mathcal{B}} \lambda_A
$.
\end{itemize}
We verify that the resulting volume-based diversity measure
$\delta_\cV$ coincides with $v_\lambda$. To this end,
let $S \subseteq X$. Then
$v_\lambda(S) = \sum_{A \subseteq X: A\cap
S \neq \emptyset} \lambda_A$.  On the other hand, $\delta_\cV (S) =
\mu (\bigcup_{x \in S} \{A \mid x \in A\})$ 
and $ \bigcup_{x \in S} \{A \mid x \in A\} =
\{A \subseteq X \mid \exists x \in S \mbox{ with } x \in A\} 
= \{A \subseteq X \mid A \cap S \neq \emptyset\}$.
By our definition of $\mu$, we thus get
$\delta_\cV (S) = \mu(\{A \subseteq X \mid A \cap S \neq \emptyset\}) =
\sum_{A \subseteq X: A\cap S \neq \emptyset} \lambda_A = v_\lambda(S)$.

For the other direction, suppose that we are given a volume-based
diversity function $\delta_{\cV}$ with $\cV =(\cS, \mu, \beta)$.  From
this, we define $\lambda$ as follows:

$$\lambda_A = \mu\big(\bigcap_{a\in A} \beta(a) \setminus 
\bigcup_{x \in X \setminus A} \beta(x) \big)$$

Intuitively, we define the ``contribution'' of $A$ to the overall
diversity as the volume shared by all balls contributing to $A$ minus
those parts of the $\sigma$-algebra which are covered by the balls
corresponding to the complement of $A$.

It remains to show that $v_\lambda(S) = \delta_\cV(S)$ holds for every
$S \subseteq X$. 
Consider an arbitrary subset $S\subseteq X$. Then the following 
chain of equalities holds: 
\begin{align}
    \delta_\cV(S) &= \mu(\bigcup_{s\in S} \beta(s) )\\
    & = \mu\bigg(\bigcup_{\substack{A\cup B = X\\ A\cap B = \emptyset}}\Big(\big((\bigcup_{s\in S} \beta(s))\cap (\bigcap_{a\in A} \beta (a))\big)\setminus (\bigcup_{b\in B}\beta (b))\Big) \bigg)\\
    & = \sum_{\substack{A\cup B = X\\ A\cap B = \emptyset}}\mu\bigg(\Big((\bigcup_{s\in S} \beta(s))\cap (\bigcap_{a\in A} \beta (a))\Big)\setminus (\bigcup_{b\in B}\beta (b))\bigg)\\
    & = \sum_{\substack{A\cup B = X\\ A\cap B = \emptyset}}\mu\bigg(\Big((\bigcup_{s\in S\setminus B} \beta(s))\cap (\bigcap_{a\in A} \beta (a))\Big)\setminus (\bigcup_{b\in B}\beta (b))\bigg)\\
    & = \sum_{\substack{A\cup B = X\\ A\cap B = \emptyset}}\begin{cases}
        \mu\bigg(\Big((\bigcap_{a\in A} \beta (a))\Big)\setminus (\bigcup_{b\in B}\beta (b))\bigg) & S\setminus B \neq \emptyset\\
        \mu\bigg(\Big((\emptyset)\cap (\bigcap_{a\in A} \beta (a))\Big)\setminus (\bigcup_{b\in B}\beta (b))\bigg) & S\setminus B = \emptyset
    \end{cases}\\
    & = \sum_{\substack{A\cup B = X\\ A\cap B = \emptyset\\ A\cap S \neq \emptyset}}\mu\bigg(\Big((\bigcap_{a\in A} \beta (a))\Big)\setminus (\bigcup_{b\in B}\beta (b))\bigg)\\
    & = \sum_{\substack{A\cup B = X\\ A\cap B = \emptyset\\ A\cap S \neq \emptyset}}\lambda_A\\
    & = \sum_{A\subseteq X\colon A\cap S \neq \emptyset}\lambda_A\\
    & = v_\lambda(S)
\end{align}

\noindent
The correctness of the above equalities is seen as follows:
First, to see that 
$$(\bigcup_{s\in S} \beta(s)) = \bigcup_{\substack{A\cup B = X\\ A\cap B = \emptyset}}
\Big(\big((\bigcup_{s\in S} \beta(s))\cap (\bigcap_{a\in A} \beta (a))\big)\setminus (\bigcup_{b\in B}\beta (b))\Big)$$
holds, consider an arbitrary element $e\in (\bigcup_{s\in S} \beta(s))$ in the set on the left-hand side.
To show that it is also contained in the set on the right-hand side of the equality, 
we define the set $A_e\subseteq X$ as $A_e = \{ a \in X \mid e \in \beta(a)\}$.
By this definition, we clearly have $e\in  (\bigcap_{a\in A_e} \beta (a)) $.
Moreover, for $B_e = X \setminus A_e$, we have that, for all $b \in B_e$, 
$e \not\in \beta(b)$ holds. Hence, $e\not \in (\bigcup_{b\in B_e}\beta (b))$.
Since we started with the assumption that $e\in (\bigcup_{s\in S} \beta(s))$ holds,
we indeed have that $e$ is contained  in the set on the right hand side, i.e., we have proven ``$\subseteq$''.
The other inclusion, i.e., ``$\supseteq$'' trivially holds as the right hand side 
is a union of subsets of $(\bigcup_{s\in S} \beta(s))$.
This proves the  Equation  (3) to (4).

Next, we show the disjointness of the union on the right-hand side.
That is, we have to show that every element $e$ contained in the set of the right-hand side
is contained in 
$$\Big(\big((\bigcup_{s\in S} \beta(s))\cap (\bigcap_{a\in A} \beta (a))\big)\setminus (\bigcup_{b\in X \setminus A}\beta (b))\Big)$$
for exactly one $A$. We now show that this unique $A$ is the set $A_e$ defined before. 
Assume to the contrary that $e$ is contained in 
$$\Big(\big((\bigcup_{s\in S} \beta(s))\cap (\bigcap_{a\in A} \beta (a))\big)\setminus (\bigcup_{b\in X \setminus A}\beta (b))\Big)$$
for some $A \subseteq X$ with $A \neq A_e$. Then either $A\setminus A_e \neq \emptyset$ or
$A_e\setminus A \neq\emptyset$. In the first case, consider an $a\in A\setminus A_e$.
Then $e\not \in \beta(a)$ and, therefore, $e\not \in (\bigcap_{a\in A} \beta (a))$,
which contradicts 
$$e\in \big((\bigcup_{s\in S} \beta(s))\cap (\bigcap_{a\in A} \beta (a)) \setminus (\bigcup_{b\in B}\beta (b)).$$
In the second case, consider an element $b\in A_e\setminus A$.
Then, $b\in B_e$ and $e\in \beta(b)$ as well as $e\in(\bigcup_{b\in B}\beta (b))$,
which again contradicts 
$$e\in \big((\bigcup_{s\in S} \beta(s))\cap (\bigcap_{a\in A} \beta (a)) \setminus (\bigcup_{b\in B}\beta (b)).$$
This proves the  Equation  (4) to (5).

In Equation (5) to (6), we remove the sets $\beta(b)$ for $b\in S\cap B$ from $(\bigcup_{s\in S} \beta(s))$ as these are subtracted at the end by $(\bigcup_{b\in B}\beta (b))$ anyway.

In Equation (6) to (7), notice that $S\setminus B\subseteq A$.
Thus, as long as there is an $a\in S\setminus B\subseteq A$, we have $(\bigcap_{a\in A} \beta (a))\subseteq \beta (a) \subseteq (\bigcup_{s\in S\setminus B} \beta(s))$
Thus, we simply replace $\Big((\bigcup_{s\in S\setminus B} \beta(s))\cap (\bigcap_{a\in A} \beta (a))\Big)$ by $\emptyset$ or $ (\bigcap_{a\in A} \beta (a))$ depending on whether $S\setminus B =\emptyset$.

In Equation (7) to (8), we simply ignore the 0's (i.e., the terms $\mu(\emptyset)$) in the sum.
The remaining equalities are straightforward applications of definitions.
\end{proof}

	\section{Additional Details for Section \ref{sec:DiversityDistance}}
    \label{app:sec:distance}
\subsection{Proof of Proposition \ref{prop:distance-no-properties}}

\propdistancenoproperties*

\begin{proof}
    Simply consider the tuples 
    $$t_1 = R(a,b,c,d,e), t_2 = R(a,b,f,g,h),t_3=R(x,y,i,j,k),t_4=(x,y,l,m,n)$$
    with Hamming distances $d(t_1,t_2) = d(t_3,t_4) = 3$ and $d(t_1,t_3)=d(t_1,t_4) = d(t_2,t_3)=d(t_2,t_4)=5$.
    Then, 
    \begin{align*}
        \deltasum&(\{t_1,t_2,t_3,t_4\})-\deltasum(\{t_2,t_3,t_4\})= 2 (d(t_1,t_2) + d(t_1,t_3) + d(t_1,t_4)) =26 \\
        &> \deltasum(\{t_1,t_2,t_3\})-\deltasum(\{t_2,t_3\})= 2(d(t_1,t_2) + d(t_1,t_3)) =16, 
    \end{align*}
    and 
    \begin{align*}
        \deltamin&(\{t_1,t_2,t_3,t_4\})-\deltamin(\{t_2,t_3,t_4\})= 3-3 = 0 \\
        &> \deltamin(\{t_1,t_2,t_3\})-\deltamin(\{t_2,t_3\})= 3-5 = -2. 
    \end{align*}
    Thus, both violate submodularity.
    Furthermore, 
    \begin{align*}
        \deltamin(\{t_2,t_3\}) = 5 > \deltamin(\{t_1,t_2,t_3,t_4\}) = 3. 
    \end{align*}
    Thus, $\deltamin$ vioates monotonicity.
\end{proof}

\subsection{Proof of Theorem \ref{theo:no-distance}}

\thmnodistance*

\begin{proof}
	We can show this using $\Velem$. To that end, note that $\Velem$ is oblivious to constant names.
	Then, consider the sets of triples 
	\begin{align*}
		S_1&=\{t_1 = R(a,b,x), \quad t_2 = R(a,y,c), \quad t_3 = R(z,b,c)\}\\
		S_2&=\{t_4 = R(a,b,x), \quad t_5 = R(a,c,y), \quad t_6 = R(a,d,z)\}\\
		S_3&=\{t_7 = R(b,a,x), \quad t_8 = R(c,a,y), \quad t_9 = R(d,a,z)\}\\
		S_4&=\{t_{10} = R(b,x,a), \quad t_{11} = R(c,y,a), \quad t_{12} = R(d,z,a)\}.
	\end{align*}
	Then, $\delta_{\Velem}(S_1)=6$ while $\delta_{\Velem}(S_2) = \delta_{\Velem}(S_3) = \delta_{\Velem}(S_4) = 7$.
	Now let $d, \mathit{agg}$ be as required and assume $\delta_{\mathit{agg},d}(S_2) = \delta_{\mathit{agg},d}(S_3) = \delta_{\mathit{agg},d}(S_4)$.
	Then, observe that due to $d$ being oblivious to the names of constants. 
	\begin{align*}
		d(t_{1},t_{2})&=d(t_4,t_5) = d(t_4,t_6) = d(t_5,t_6),\\
		d(t_1,t_3)&=d(t_7,t_8)=d(t_7,t_9)=d(t_8,t_9),\\
		d(t_2,t_3)&=d(t_{10},t_{11})=d(t_{10},t_{12})=d(t_{11},t_{12})
	\end{align*}
	W.l.o.g., $$d(t_{1},t_{2})\leq d(t_1,t_3) \leq d(t_2,t_3)$$
	Thus, due to monotonicity
	$$ \delta_{\mathit{agg},d}(S_2)\leq  \delta_{\mathit{agg},d}(S_1)\leq  \delta_{\mathit{agg},d}(S_4)$$
	Thus, $ \delta_{\mathit{agg},d}(S_1)= \delta_{\mathit{agg},d}(S_2)$ and, e.g.,  $\mathit{agg},d$ cannot distinguish $S_1$ from $S_2$
	
\end{proof}

\subsection{Proof of Proposition \ref{thm:weitzman:submodular}}
\label{app:TheoremWeitzmanSubmodularity}

\thmweitzmansubmodular*

\begin{proof}
Consider the universe $\U = \{a,b,c,d\}$ with the distance function $d$ 
shown in Table~\ref{tab:distance}.

\begin{table}[h!]
\centering
\begin{tabular}{|c|c|c|c|c|}
\hline
   & $a$ & $b$ & $c$ & $d$ \\
\hline
$a$ & 0 & 2 & 2 & 2 \\
\hline
$b$ & 2 & 0 & 1 & 1 \\
\hline
$c$ & 2 & 1 & 0 & 2 \\
\hline
$d$ & 2 & 1 & 2 & 0 \\
\hline
\end{tabular}
\caption{Distance function $d$ in the proof of Theorem \ref{thm:weitzman:submodular}}
\label{tab:distance}
\end{table}

\noindent
It is easy to verify that $d$ is actually a metric, i.e., 
it satisfies the following properties: 

\begin{enumerate}[label=(\arabic*)]
    \item non-negativity: $d(x,y) \geq 0 $ for every pair $x,y \in \U$;
    \item symmetry: $d(x,y) = d(y,x)$  for any two elements
    $x,y \in \U$; 
    \item $d(x,x) = 0$ for every $x \in \U$;
    \item $d$ satisfies the triangle inequality, i.e., 
    $d(x,z) \leq d(x,y) + d(y,z)$ for any three (pairwise distinct) elements 
    $x,y,z \in \U$. This clearly holds, since we have $d(x,y) + d(y,z) \geq 2$ and 
    $d(x,z) \leq 2$ for any three pairwise distinct elements 
    $x,y,z \in \U$.   
    \item $d(x,y) = 0$ implies $x = y$ for all $x,y \in \U$. 
\end{enumerate}
 
\smallskip
\noindent
If Weitzman's diversity measure $\deltaW$ is defined via this distance function $d$,
it is easy to veriy that the following equalities hold: 

\begin{center}   
$\deltaW(\{a,b\}) = 2$, 
$\deltaW(\{a,b,c\}) = 3$, 
$\deltaW(\{a,b,d\}) = 3$, and
$\deltaW(\{a,b,c,d\}) = 5$.
\end{center}

We can now can show that there exist 
sets $S_1,S_2 \subseteq \U$, such that 
$\deltaW(S_1 \cup S_2) + \deltaW(S_1 \cap S_2)
\geq \deltaW(S_1) + \deltaW(S_2).$
by setting
$S_1 = \{a,b,c\}$ and 
$S_2 = \{a,b,d\}$. Then we have: 

\begin{center}
$\deltaW(\{a,b,c,d\}) + \deltaW(\{a,b\})=  5 + 2 = 7$
and
$\deltaW(\{a,b,c\}) + \deltaW(\{a,b,d\})=  3+3 = 6$ \\  
\end{center}

\noindent That is, this particular diversity measure $\deltaW$ violates submodularity.
\end{proof}

\subsection{Proof of Theorem \ref{thm:volume:vs:weitzman}}
\label{app:TheoremWeitzman}

\thmvolumevsweitzman*

\begin{proof}
First, recall that an ultrametric is a distance function 
satisfying the conditions (1), (2), (3), and (5) of a metric 
recalled in Section \ref{app:TheoremWeitzmanSubmodularity}
plus (4') the strong triangle inequality: 
$d(a,c) \leq \max \big( d(a,b), d(a,c) \big)$.
In case of an ultrametric on $\U$, the elements of 
$\U$ can be arranged at the leaf nodes of a 
hierarchical (= ``taxonometric'') tree $T$, where $T$ is obtained as follows: With each inner node
of $T$, we can associate the radius of the ball that contains all elements at descendants of this node.
In particular, with the root node of $T$, we associate the radius $r$ of 
(the ball containing all of) 
$\U$, i.e., $r = \max \{ d(a,b) | a,b \in \U \}$. 
Moreover, the length of the edges can be chosen in such a way that $d(a,b)$ for two distinct elements $a,b \in \U$ 
is equal to the path length from each of these nodes to their nearest common ancestor. Note that then 
the path from the root to any leaf node has length $r$.
Now, if the distance function is an ultrametric, then 
the Weitzman diversity $\delta_W(S)$ of a subset 
$S \subseteq \U$ can be defined as follows: let $T'$ be the smallest subtree of the hierarchical tree $T$ 
that contains the root of $T$ and all leaf nodes corresponding to the 
elements in $S$. Moreover, let $r$ denote the radius of $\U$. Then $\delta_W(S)$ is equal to the 
sum of the edge-lengths in $T'$ minus $r$. 

We now define 
a volume assignment $\cV = (\cS, \mu, B)$ by defining $\cS$ as the 
set of edges in the hierarchical tree $T$, 
$\beta$ maps every element $j$ of $\U$ to the set of edges of the path from the root to the node corresponding to $j$ in $T$, 
and $\mu(S)$ is defined as the sum of the lengths of the edges in $S$. 
Then, for every non-empty subset $S \subseteq \U$, the resulting volume-based diversity
function 
$\delta_\cV$ satisfies $\delta_\cV(S) = \delta_W(S) + r$, where $r$ is the radius of $\U$.
\end{proof}

\section{Additional Details for Section~\ref{sec:Exact}} 

\subsection{Proof of Theorem~\ref{theo:CQ-hardness}}

\theoCQhardness*

\begin{proof}

We provide a reduction for the problems individually.
However, we always use the $\IndependentSet$  problem as the basis for the reductions.
To that end, recall that in ~\cite{DBLP:conf/ciac/AlimontiK97}, it was shown that the problem remains $\NP$-hard when restricted to $3$-regular graph. 
That is, we consider instances $(V,E,k)$ where $G=(V,E)$ is a graph where every vertex $v\in V$ has exactly 3 incident edges and $k\in \mathbb{N}$.
Then, $(V,E,k)$ is a yes-instance iff there exists a set of pairwise non-adjacent vertices $I\subseteq V$ such that $|I|=k$.
We acknowledge that the provided reductions are similar to those given in \cite{MerklPS23}.

\textbf{Reduction for $\Velem$:}
For the reduction, let us consider a database $D$ encoding the incident relation of the undirected graph $G=(V,E)$.
That is, for every $u\in V$ and its 3 incident edges $uv_1,uv_2,uv_3\in E$ ($v_i\neq v_j$ for $i\neq j$), the database contains the tuple $Inc(u, uv_1,uv_2,uv_3)$.
However, to only include one atom per vertex, let us consider an arbitrary order $\preceq$ over $V$.
Then, we define
\begin{align*}
    D := \{Inc(u,uv_1,uv_2,uv_3)\mid uv_1,uv_2,uv_3\in E \text{ and } v_1 \prec v_2 \prec v_3\}
\end{align*}
Furthermore, consider the query $Q$ defined independent of $G$:
\begin{align*}
    Q(x,y_1,y_2,y_3) \leftarrow Inc(x,y_1,y_2,y_3)
\end{align*}
Consequently, $\sem{Q}(D)$ consists exactly of tuple $Q(u, uv_1,uv_2,uv_3)$ where $u\in V$ and $uv_1,uv_2,uv_3$ are its 3 incident edges.
Then, for a subset $S\subset \sem{Q}(D)$ of size $k$, the value $\delta_{\Velem}(S)$ equals the number of different vertices $v\in V$ used in $S$ plus the number of unique incident edges over all $v$ used.
Thus, $\delta_{\Velem}(S)$ is at most $4k$ and it is $4k$ exactly when $S$ uses $k$ distinct vertices $u_1,\dots,u_k\in V$ such that for no two $i,j=1,\dots,k, i\neq j$ the vertices $u_i$ and $u_j$ have a common incident edge.
Consequently, $\{u_1,\dots,u_k\}$ is an independent set of $G$.

Conversely, when $\{u_1,\dots,u_k\}$ is an independent set of $G$ of size $k$, then for no two $i,j=1,\dots,k, i\neq j$ the vertices $v_i$ and $v_j$ have a common incident edge.
Consequently, 
$$\delta_{\Velem}(\{Q(u_1,u_1v_{1,1},u_1v_{1,2},u_1v_{1,3}), \dots, Q(u_k,u_kv_{k,1},u_kv_{k,2},u_kv_{k,3})\})=4k.$$
This completes the reduction for $\Velem$.

\textbf{Reduction for $\Vpos$:}
For $\Vpos$, we have to adapt the previous reduction.
The problem with using the reduction as is, is that $\Vpos$ counts the same edge twice when they appear at different positions of the tuples.
Thus, essentially, we have to make sure that edges $uv$ appear at the same position for both incident vertices $u$ and $v$.
To do so, we simply increase the arity of $Inc$ to 6.
Then, the database $D$ encoding $G$ will be defined via a iterative process.
To that end, we start with the initial database 
$D_0:=\{Inc(u,u,u,u,u,u) \mid u\in V\}$
Then, we iterate through the edges $\{e_1,\dots,e_{|E|}\}=E$.
Ad each step, we consider a $u_iv_i=e_i$ and define the database $D_i$ as follows:
\begin{itemize}
    \item Let $Inc(u_i,c_1,c_2,c_3,c_4,c_5) Inc(v_i,c'_1,c'_2,c'_3,c'_4,c'_5)\in D_{i-1}$ be the tuple associated to $u_i$ and $v_i$, respectively.
    \item Then, by the pigeonhole principle and since $u_i$ and $v_i$ both have only 3 incident edges, there must be a $j\in \{1,\dots,5\}$ such that $c_j,c'_j$ are as initialized, i.e., $c_j=u_i$ and $c'_j=v_i$.
    \item We replace both $c_j,c'_j$ with $e_i$. E.g., for $j=5$,
    $$D_{i}:= D_{i-1}\setminus \{Inc(u_i,c_1,\dots,c_5) Inc(v_i,c'_1,\dots,c'_5)\} \cup \{Inc(u_i,c_1,c_2,c_3,c_4,e_i) Inc(v_i,c'_1,c'_2,c'_3,c'_4,e_i)\}$$
\end{itemize}

Furthermore, consider the query $Q$ defined independent of $G$:
\begin{align*}
    Q(x,y_1,y_2,y_3,y_4,y_5) \leftarrow Inc(x,y_1,y_2,y_3,y_4,y_5)
\end{align*}

Consequently, $\sem{Q}(D)$ consists exactly of tuple $Q(u, e_1,e_2,e_3,e_4,e_5)$ where $u\in V$ and $\{e_1,\dots,e_5\}\cap E$ are the 3 edges incident to $u$ and $\{e_1,\dots,e_5\}\setminus E = \{u\}$.
Furthermore, for two tuples 
$$Q(u, e_1,e_2,e_3,e_4,e_5), Q(u', e'_1,e'_2,e'_3,e'_4,e'_5)\in \sem{Q}(D)$$
with a common constant $c$ on a common position, it must be on one of last 5 positions and equal to $uu'\in E$.
The converse is also true.
Then, for a subset $S\subset \sem{Q}(D)$ of size $k$, the value $\delta_{\Vpos}(S)$ equals 3 times the number of different vertices $v\in V$ used in $S$ plus the number of unique incident edges over all $v$ used.
Thus, $\delta_{\Vpos}(S)$ is at most $6k$ and it is $6k$ exactly when $S$ uses $k$ distinct vertices $u_1,\dots,u_k\in V$ such that for no two $i,j=1,\dots,k, i\neq j$ the vertices $u_i$ and $u_j$ have a common incident edge.
Consequently, $\{u_1,\dots,u_k\}$ is an independent set of $G$.

Conversely, when $\{u_1,\dots,u_k\}$ is an independent set of $G$ of size $k$, then for no two $i,j=1,\dots,k, i\neq j$ the vertices $v_i$ and $v_j$ have a common incident edge.
Consequently, 
$$\delta_{\Vpos}(\{Q(u_1,e_{1,1},e_{1,2},e_{1,3},e_{1,4},e_{1,5}, \dots, Q(u_k,e_{k,1},e_{k,2},e_{k,3},e_{k,4},e_{k,5}\})=6k.$$
This completes the reduction for $\Vpos$.

\textbf{Reduction for $\Velemw$ and $\Vposw$:}
Actually, in both cases, we can simply reuse the reduction for $\Velem$ and $\Vpos$, respectively, by simply assigning all constants the weight $1$.

\textbf{Reduction for $\VQD$:}
For $\VQD$, we have to adapt the reduction given for $\Velem$.
To that end, we simply extend the database $D$ considered there by the atoms $E(uv)$ for edges $uv\in E$.
Then, consider the query $Q$ defined independent of $G$:
\begin{align*}
    Q(x,y_1,y_2,y_3) \leftarrow Inc(x,y_1,y_2,y_3), E(y_1), E(y_2), E(y_3)
\end{align*}
Consequently, $\sem{Q}(D)$ consists exactly of tuple $Q(u, uv_1,uv_2,uv_3)$ where $u\in V$ and $uv_1,uv_2,uv_3$ are its 3 incident edges such that $v_1\prec v_2\prec v_3$ and $$\betaQD(Q(u, uv_1,uv_2,uv_3)) = \{Inc(v,uv_1,uv_2,uv_3), E(uv_1), E(uv_2), E(uv_3)\}.$$
Of the tuples $\betaQD(Q(u, uv_1,uv_2,uv_3))$, the first is unique while the remaining ones are shared with the tuple corresponding to the other incident vertex $v_i,i=1,2,3$.
Thus, for a subset $S\subset \sem{Q}(D)$ of size $k$, the value $\delta_{\VQD}(S)$ equals the number of different vertices $v\in V$ used in $S$ plus the number of unique incident edges over all $v$ used.
Thus, $\delta_{\Velem}(S)$ is at most $4k$ and it is $4k$ exactly when $S$ uses $k$ distinct vertices $u_1,\dots,u_k\in V$ such that for no two $i,j=1,\dots,k, i\neq j$ the vertices $u_i$ and $u_j$ have a common incident edge.
Consequently, $\{u_1,\dots,u_k\}$ is an independent set of $G$.

Conversely, when $\{u_1,\dots,u_k\}$ is an independent set of $G$ of size $k$, then for no two $i,j=1,\dots,k, i\neq j$ the vertices $v_i$ and $v_j$ have a common incident edge.
Consequently, 
$$\delta_{\Velem}(\{Q(u_1,u_1v_{1,1},u_1v_{1,2},u_1v_{1,3}), \dots, Q(u_k,u_kv_{k,1},u_kv_{k,2},u_kv_{k,3})\})=4k.$$
This completes the reduction for $\Velem$.
\end{proof}

\subsection{Proof of Theorem~\ref{theo:CQ-approximation-hardness}}

\theoCQapproximationhardness*

\begin{proof}

We will prove this by encoding the $\maxCoverage$ problem.
To that end, recall the definition of $\maxCoverage$ (where $\Omega$ is simply a countable infinite universe):
\begin{center}
		\begin{tabular}{rl}
			\hline\\[-2ex]
			\textbf{Problem:} & $\maxCoverage$\\
			\textbf{Input:} & A set of finite subsets $\mathcal{C} =\{S_1,\dots,S_m\} \subseteq 2^{\Omega}$ and a $k\geq 0$. \\
			\textbf{Output:} & $\argmax_{\mathcal{C}'\subseteq \mathcal{C}, |\mathcal{C}'|=k} |\bigcup_{S\in \mathcal{C}'}S|$\\\\[-2.2ex]
			\hline
		\end{tabular}
\end{center}

Further, recall that due to \cite{DBLP:journals/jacm/Feige98}, it is know that unless $\PTIME = \NP$, there is no algorithm to approximate $\maxCoverage$ in polynomial time.

We can construct very naturally a volume assignment $\cV=(\cS,\mu,\beta)$ such that for a simply CQ $Q$ the problem $\CQEval[\Sigma,\cV,Q]$ encodes $\maxCoverage$.
To do so, we only use one relation symbol $\Sigma=\{Q,R_{\mathcal{C}}\}$ and data values $\mathbb{D}:= \{S\subseteq \Omega \mid S \text{ is finite}\}$.
Then, the universe of tuples $\TUPLES$ is 
$\TUPLES = \{R_{\mathcal{C}}(S)\mid S\subseteq \Omega, S \text{ is finite}\}$.
We then define the volume assignment via
\begin{align*}
    \cS := \{S\subseteq \Omega \mid S \text{ is finite}\},  \quad \mu=\mucount, \quad \beta\colon R_{\mathcal{C}}(S)\mapsto S, Q(S)\mapsto S
\end{align*}
That is, the data values are finite subsets of the universe and intuitively, we measure the diversity of a set of subsets by the number of distinct elements covered.
Further, we use the single atom query 
$$Q(x)\leftarrow R_{\mathcal{C}}(x)$$
and encode an instance $(\mathcal{C},k)$ in the database 
$$D:=\{R_{\mathcal{C}}(S)\mid S\in \mathcal{C}\}$$
Then, clearly, 
$$\delta_{\cV}(\{Q(S_1), \dots, Q(S_k)\}) = \delta_{\cV}(\{R_{\mathcal{C}}(S_1), \dots, R_{\mathcal{C}}(S_k)\}) = |\bigcup_{i=1,\dots,k}S_i|$$
and $\{Q(S_1), \dots, Q(S_k)\}\subseteq \sem{Q}(D)$ iff $\{S_1,\dots, S_k\}\subseteq \mathcal{C}$.
Hence, finding maximal diverse solutions of $Q$ over these databases is equivalent to finding subsets of $\mathcal{C}$ with maximal coverage.    
\end{proof}

        \section{Additional Details for Section \ref{sec:ACQs}} \label{app:acq}

\subsection{Proof of Theorem \ref{thm:impapprox}}

\thmimpapprox*

\begin{proof}%
    This theorem can essentially be proven in the same way as Theorem \ref{theo:CQ-approximation}.
    That is, it immediately follows from the results of~\cite{nemhauser1978analysis} as the considered diversity functions $\delta_{\cV}$ are submodular.
    We simply have to notice that the runtime of Algorithm~\ref{alg:greedy} with the additional input $Q$ is as claimed as $T$, in particular, accounts for the time spent in line \ref{line:greedy}.
\end{proof}

\subsection{Proof of Theorem \ref{theo:hardnessCombinedComplexity}}

\theoHardnessCompbinedComplexity*

\begin{proof}
We provide a reduction for the problems individually.
However, we always use the $\HamiltonianPath$  problem as the basis for the reductions.
Recall that a Hamiltonian path of a directed graph $G=(V,E)$ is a sequence of distinct vertices $(v_1,\dots,v_n)$ such that $v_iv_{i+1}\in E$ for $i=1,\dots,|V|-1$
Then, $(V,E)$ is a yes-instance iff there exists a Hamiltonian path of $(V,E)$.

\textbf{Reduction for $\Velem$:}
Given an instance $G = (V(G), E(G))$ of $\HamiltonianPath$, we define an instance $(D,Q,S)$ of $\CQNext[\Sigma, \cV]$
as follows: the database $D$ consists of a single binary relation $E$ storing the edges of $G$.
That is,
$$D=\{E(u,v)\mid uv\in E(G)\}$$
Further, we set $ S= \emptyset$,
and, for $n = |V(G)|$, we define the ACQ $Q$ as follows:
        \[
        Q(x_1,\dots,x_{n}) \ \leftarrow \  E(x_1,x_2),\dots,E(x_{n-1},x_{n}).
        \]
        Notice that a solution $Q(h(\overline{x}))\in \sem{Q}(D)$ corresponds to the walk $h(\bx)$ (a sequence of not necessarily distinct vertices) in $G$.
        Applying $\delta_{\Velem}$ to the singleton set $\{Q(h(\overline{x}))\}$, i.e., $\delta_{\Velem}(\{Q(h(\overline{x}))\})$, counts the number of distinct vertices used in the corresponding walk. 
        Hence, $G$ is a positive instance of Hamiltonian path, if and only if 
        the solution to this instance of $\CQNext[\Sigma, \Velem]$ yields an answer $Q(h(\overline{x}))$ with 
        $\delta_{\Velem}(\{Q(h(\overline{x}))\}) = n$.

\textbf{Reduction for $\Velemw$:}
We can simply reuse the reduction for $\Velem$ by simply assigning all constants the weight $1$.

\textbf{Reduction for $\VQD$:}
Given an instance $G = (V, E)$ of $\HamiltonianPath$, we define an instance $(D,Q,S)$ of $\CQNext[\Sigma, \cV]$
as follows: the database $D$ consists of the binary relation $E$ storing the edges of $G$ as well as a relation $V$ storing the vertices.
That is, 
$$D=\{E(u,v)\mid uv\in E\}\cup\{V(v)\mid v\in V(G)\}$$
Further, we set $ S= \emptyset$,
and, for $n = |V(G)|$, we define the ACQ $Q$ as follows:
        \[
        Q(x_1,\dots,x_{n}) \ \leftarrow \  E(x_1,x_2),\dots,E(x_{n-1},x_{n}), V(x_1),\dots, V(x_n).
        \]
        Notice that a solution $Q(h(\overline{x}))\in \sem{Q}(D)$ corresponds to the walk $h(\bx)$ (a sequence of not necessarily distinct vertices) in $G$.
        Applying $\delta_{\VQD}$ to the singleton set $\{Q(h(\overline{x}))\}$, i.e., $\delta_{\VQD}(\{Q(h(\overline{x}))\})$, counts the number of distinct vertices and edges used in the corresponding walk. 
        Hence, $G$ is a positive instance of Hamiltonian path, if and only if 
        the solution to this instance of $\CQNext[\Sigma, \Velem]$ yields an answer $Q(h(\overline{x}))$ with 
        $\delta_{\Velem}(\{Q(h(\overline{x}))\}) = 2n-1$.
\end{proof}

\subsection{Proof of Theorem \ref{theo:tractableCombinedComplexity}}

\theoTractableCompbinedComplexity*

\begin{proof}
Since $\Vposw$ is strictly more general than $\Vpos$ we simply provide the proof for $\Vposw$.
That is, we show that $\CQNext[\Sigma, \Vposw]$ is tractable for ACQs.
To that end, 
let
$$Q(\bx) \ \leftarrow \ R_1(\bx_1), \ldots R_{m}(\bx_{m})$$
be a ACQ 
over some schema $\Sigma$, 
$D$ a database, %
and $h_1,\dots,h_k$ homomorphisms from $Q$ to $D$, i.e., $S=\{Q(h_1(\bx)),\dots, Q(h_k(\bx))\}\subseteq \sem{Q}(D)$ a $k$-set of solutions.
Further, let $\bx=(x_1,\dots,x_{|\bx|})$.
Consider the \textit{marginal} diversity for a new solution $Q(h(\bx))\in \sem{Q}(D)$:
\[
\delta_{\Vpos}(S\cup \{Q(h(\bx))\})-\delta_{\Vpos}(S) \ = \  \sum_{x_l\in \bx} \alpha_{x_l}, \quad \text{where }\alpha_{x_l}=\begin{cases}
w(h(x_l),l) & \text{if } \forall i \colon h(x_l)\neq h_i(x_l)\\
0 & \text{if } \exists i \colon h(x_l) = h_i(x_l)\\
\end{cases}
\]
namely, the marginal diversity counts the number of \textit{new} values (at new positions).
We can cast this then to a sum-product query over the tropical semi-ring $\mathbb{R}_{\max}:=(\mathbb{R}\cup\{+\infty\}, +,\max)$.
Doing so shows that we can find the element that maximizes the marginal diversity in linear time.
To do so, for every $x_l\in \bx$ let us choose a covering relation $R^{x_l}:=R_i$ where $x_l\in \bx_i$ and $1,\dots,m$.
Then, we can define the $\mathbb{R}_{\max}$-annotated relations $R^*_1,\dots,R^*_m$.
That is, for $h\colon \mathcal{X}\rightarrow \D$ such that $R_j(h(\bx_j))\in D$ we add to $D$ the annotated tuple
$$R_j^*(h(\bx_j))\mapsto  \sum_{{x_l}\in \bx\colon R^{x_l}=R_j} \alpha_{x_l}, \quad \text{where }\alpha_{x_l}=\begin{cases}
w(h(x_l),l) & \text{if } \forall i=1,\dots,k \colon h(x_l)\neq h_i(x_l),\\
0 & \text{if } \exists i=1,\dots,k \colon h(x_l) = h_i(x_l). \\
\end{cases}$$
Note that we only use information from $h(\bx_j)$ and the relation name $R_j$.
Then, as every variable $x\in \bx$ is covered by exactly one relation, we have: 
\begin{align}
    \delta_{\Vpos}(S\cup \{Q(h(\bx))\})-\delta_{\Vpos}(S) = \sum_{j=1}^{m} R_m^*(h(\bx_m)), \label{eq:app:spq}
\end{align}
for homomorphism $h$ form $Q$ to $D$ where we add up the annotations on the right hand side.
Note that we can see minimizing the right hand side of Eq. \eqref{eq:app:spq} as a sum-product query of the following form:
\begin{align*}
    Q^*()\leftarrow \bigotimes_{j=1}^m R_m^*(\bx_m)
\end{align*}
With the semantics 
$$\sem{Q^*()}(D)=\bigoplus_{h\colon \text{ hommomorphism from } Q \text{ to } D}\bigotimes_{j=1}^m R_m^*(h(\bx_m))$$
Note that over the tropical semi-ring $\oplus = \max$ while $\otimes = +$.
Thus, actually,
\begin{align*}
    \sem{Q^*()}(D) = \max_h \sum_{j=1}^m R_m^*(h(\bx_m)).
\end{align*}
Then, due to results on sum-product queries~\cite{DBLP:conf/pods/KhamisNR16,DBLP:journals/jcss/PichlerS13}
we can determine the value of $\sem{Q^*()}$ as well as a $h$ witnessing this maximum in time $O(|Q|\cdot |D|)$.
Note that $Q(h(\bx))\in \sem{Q}(D)$ finding $h$ is possible by tracing how $\sem{Q^*()}$ is computed as it is always only necessary to trace one argument of $\oplus=max$.
Consequently, solving $\CQNext[\Sigma, \Velem]$ for ACQs is possible in time $O(|Q|\cdot |D|)$.
Then, due to Theorem \ref{thm:impapprox}, we can compute a $(1-\nicefrac{1}{e})$-approximation of $\CQEval[\Sigma, \Velem]$ in time $O(k\cdot|Q|\cdot |D|)$.
\end{proof}

\subsection{Proof of Theorem \ref{thm:appeval}}

To prove Theorem \ref{thm:appeval}, we first recall the notions used in \cite{deep2025ranked}:

\newcommand{\rank}{\texttt{rank}}

\begin{definition}[\cite{deep2025ranked}]
\label{def:rankOrder}
Let $\rank$ be a ranking function ($\mathbb{T}_{\{R\}}\rightarrow \mathbb{R}$) over $R(\bx)$-tuples and $\by \subseteq \bx$. 
In the following, $g$ is a arbitrary homomorphism over $\bx\setminus\by$.
We say that $\rank$ is {\em $\by$-decomposable} 
if there exists a total order $\succeq$ for all homomorphisms over $\by$, such that for any two homomorphisms $h,h'$ over $\by$ we have:

$$ h \succeq h' \Rightarrow 
\begin{cases}
    \forall g, \rank(R((h\cup g)(\bx))) =  \rank(R((h'\cup g)(\bx))) \\
    \text{or} \\
    \forall g, \rank(R((h\cup g)(\bx))) >  \rank(R((h'\cup g)(\bx)))
\end{cases}
$$
\end{definition}

\begin{definition}[\cite{deep2025ranked}]
Let $\rank$ be a ranking function over $R(\bx)$-tuples, and $\by,\bz \subseteq \bx$ such that $\by\cap\bz = \emptyset$.
We say that $rank$ is {\em $\by$-decomposable conditioned on $\bz$} 
if for every homomorphism $f$ over $\bz$, the ranking function
$rank_{f}(R_{\bx\setminus \bz}(\hat{h}(\bx\setminus \bz))) := rank(R(({f \cup \hat{h}})(\bx)))$ defined over $R_{\bx\setminus \bz}(\bx\setminus \bz)$-tuples is $\by$-decomposable.
\end{definition}

\begin{definition}[\cite{deep2025ranked}]
Let $(T,\chi)$ be a tree decomposition of a full CQ $Q$. We say that a ranking function is {\em compatible
with $(T,\chi)$} if for every node $t$ it is $(\chi(T_t) \setminus \texttt{key}(t))$-decomposable conditioned on $\texttt{key}(t)$\footnote{Recall that $\chi(T_t)$ are all the variables that are in the bag of $v$ or a decedent of $v$ and $\texttt{key}(t)$ are the variables in the bag of $t$ that it shares with its parent (or $\emptyset$ if $t$ is the root).}.
\end{definition}

We connect the definitions of \cite{deep2025ranked} with the ones we give in Section \ref{sec:ACQs} next.

\begin{lemma}
\label{lem:ydec}
    Let $\beta$ be a ball function this is $\overline{y}$-decomposable conditioned on $\overline{z}$ (w.r.t.\ $R(\bx)$).
    Then, $\beta_{S}(R(\overline{x})):=\beta(R(\overline{x}))\setminus S$ is $\overline{y}$-decomposable conditioned on $\overline{z}$ (w.r.t.\ $R(\bx)$) for every $S\in \cS$. %
\end{lemma}
\begin{proof}
    Let $f$ be a homomorphism over $\overline{z}$.
    Then, let us extend $\beta_{S},\beta$ to $R_{\bx\setminus \bz}(\bx\setminus \bz)$-tuples as 
    $$\beta_{S}(R_{\bx\setminus \bz}(\hat{h}(\bx\setminus \bz))) = \beta_{S}(R((f \cup \hat{h})(\bx))), \quad \beta(R_{\bx\setminus \bz}(\hat{h}(\bx\setminus \bz))) = \beta(R((f \cup \hat{h})(\bx))).$$
    Notice, $\beta_{S} =\beta \setminus S$ also for $R_{\bx\setminus \bz}$-tuples.
    Then, let $h,h'$ be homomorphisms over $\overline{y}$, and $g,g'$ homomorphisms over $\overline{x}\setminus (\overline{y}\cup \overline{z})$.
    Then,
    \begin{align*}
        \beta_{S}&(R_{\bx\setminus \bz}((h\cup g)(\overline{x}\setminus\overline{z}))) \setminus \beta_{S}(R_{\bx\setminus \bz}((h'\cup g)(\overline{x}\setminus\overline{z}))) \\
        &=\big(\beta(R_{\bx\setminus \bz}((h\cup g)(\overline{x}\setminus\overline{z}))) \setminus S\big)\setminus \big(\beta(R_{\bx\setminus \bz}((h'\cup g)(\overline{x}\setminus\overline{z}))  \setminus S\big)\\
        &=\big(\beta(R_{\bx\setminus \bz}((h\cup g)(\overline{x}\setminus\overline{z}))) \setminus \beta(R_{\bx\setminus \bz}((h'\cup g)(\overline{x} \setminus\overline{z})))\big) \setminus S \\
        &=\big(\beta(R_{\bx\setminus \bz}((h\cup g')(\overline{x}\setminus\overline{z}))) \setminus \beta(R_{\bx\setminus \bz}((h'\cup g')(\overline{x} \setminus\overline{z})))\big) \setminus S\\
        &=\big(\beta(R_{\bx\setminus \bz}((h\cup g')(\overline{x}\setminus\overline{z}))) \setminus S\big)\setminus \big(\beta(R_{\bx\setminus \bz}((h'\cup g')(\overline{x}\setminus\overline{z}))  \setminus S\big)\\
        &=\beta_S(R_{\bx\setminus \bz}((h\cup g')(\overline{x}\setminus\overline{z}))) \setminus \beta_S(R_{\bx\setminus \bz}(h'\cup g')(\overline{x}\setminus\overline{z})))
    \end{align*}
    This completes the proof
\end{proof}

\begin{lemma}
    \label{lem:mudec}
    Let $\cV=(\cS, \mu, \beta)$ be a volume assignment over $\TUPLES$ such that $\beta$ is $\overline{y}$-decomposable conditioned on $\overline{z}$ (w.r.t.\ $R(\bx)$).
    Then, then ranking function $\rank\colon \mathbb{T}_{\{R\}}\rightarrow \mathbb{R}, R(\hat{f}(\bx))\mapsto \mu(\beta(R(\hat{f}(\bx))))$ is $\overline{y}$-decomposable conditioned on $\overline{z}$.
\end{lemma}
\begin{proof}
    Let $f$ be a homomorphism over $\overline{z}$, let $h,h'$ be homomorphisms over $\overline{y}$ and let $g$ be a homomorphism over $\overline{x}\setminus (\overline{y}\cup \overline{z})$.
    We then define 
    \begin{align*}
        h\prec h' &\text{ iff } \mu(\beta(R((h\cup g \cup f)(\overline{x})))) < \mu(\beta(R((h'\cup g \cup f)(\overline{x}))))\\
        h \equiv h' &\text{ iff } \mu(\beta(R((h\cup g \cup f)(\overline{x})))) = \mu(\beta(R((h'\cup g \cup f)(\overline{x}))))\\
        h\succ h' &\text{ iff } \mu(\beta(R((h\cup g \cup f)(\overline{x})))) > \mu(\beta(R((h'\cup g \cup f)(\overline{x}))))
    \end{align*}
    Let $g'$ be a further homomorphism over $\overline{x}\setminus (\overline{y}\cup \overline{z})$.
    Then ($\beta$ extended to $R_{\bx\setminus\bz}$ as above), 
    \begin{align*}
        \mu\big(\beta(R((h\cup g \cup f)&(\overline{x})))\big) - \mu\big(\beta(R((h'\cup g \cup f)(\overline{x})))\big)\\
        &= \mu\big(\beta(R((h\cup g \cup f)(\overline{x})))\setminus \beta(R((h'\cup g \cup f)(\overline{x})))\big) \\
        & \phantom{a} \quad\quad\quad- \mu\big(\beta(R((h\cup g \cup f)(\overline{x})))\setminus \beta(R((h'\cup g \cup f)(\overline{x})))\big)\\
        &= \mu\big(\beta(R_{\bx\setminus\bz}((h\cup g)(\overline{x}\setminus\overline{z})))\setminus \beta(R_{\bx\setminus\bz}((h'\cup g)(\overline{x}\setminus\overline{z})))\big) \\
        & \phantom{a} \quad\quad\quad - \mu\big(\beta(R_{\bx\setminus\bz}((h\cup g)(\overline{x}\setminus\overline{z})))\setminus \beta(R_{\bx\setminus\bz}((h'\cup g)(\overline{x}\setminus\overline{z})))\big)\\
        &= \mu\big(\beta(R_{\bx\setminus\bz}((h\cup g')(\overline{x}\setminus\overline{z}))\setminus \beta(R_{\bx\setminus\bz}((h'\cup g')(\overline{x}\setminus\overline{z})))\big) \\
        & \phantom{a} \quad\quad\quad - \mu\big(\beta(R_{\bx\setminus\bz}((h\cup g')(\overline{x}\setminus\overline{z})))\setminus \beta(R_{\bx\setminus\bz}((h'\cup g')(\overline{x}\setminus\overline{z})))\big) \\
        &= \mu\big(\beta(R((h\cup g' \cup f)(\overline{x})))\setminus \beta(R((h'\cup g' \cup f)(\overline{x})))\big) \\
        & \phantom{a} \quad\quad\quad - \mu\big(\beta(R((h\cup g' \cup f)(\overline{x})))\setminus \beta(R((h'\cup g' \cup f)(\overline{x})))\big) \\
        & = \mu\big(\beta(R((h\cup g' \cup f)(\overline{x})))\big) - \mu\big(\beta(R((h'\cup g' \cup f)(\overline{x})))\big)
    \end{align*}
    Thus, for every homomorphism $g'$ over $\overline{x}\setminus (\overline{y}\cup \overline{z})$
    \begin{align*}
        \rank_f(R_{\bx\setminus\bz}((h\cup g')(\overline{x}\setminus\overline{z})))&=\mu(\beta(R((h\cup g \cup f)(\overline{x})))) \\
        &= \mu(\beta(R((h'\cup g \cup f)(\overline{x})))) \\
        &= \rank_f(R_{\bx\setminus\bz}((h'\cup g')(\overline{x}\setminus\overline{z})))
    \end{align*}
    holds if $h \equiv h'$, while 
    \begin{align*}
        \rank_f(R_{\bx\setminus\bz}((h\cup g')(\overline{x}\setminus\overline{z})))&=\mu(\beta(R((h\cup g \cup f)(\overline{x})))) \\
        &> \mu(\beta(R((h'\cup g \cup f)(\overline{x})))) \\
        &= \rank_f(R_{\bx\setminus\bz}((h'\cup g')(\overline{x}\setminus\overline{z})))
    \end{align*}
    holds if $h \succ h'$.
\end{proof}

\begin{lemma}
\label{lem:tdec}
    Let $\rank$ be a ranking function over $R$-tuples that is $\by$-decomposable conditioned on $\bz$.
    Then, for any constant $c$, $\rank + c$ is also $\by$-decomposable conditioned on $\bz$.
\end{lemma}
\begin{proof}
    Proof is immediate.
\end{proof}

\begin{lemma}
\label{lem:torank}
    Let $Q(\bx)$ be a full CQ and let $\cV=(\cS, \mu, \beta)$ be a volume assignment over $\TUPLES$ such that $\beta$ is $\overline{y}$-decomposable conditioned on $\overline{z}$ (w.r.t.\ $Q$).
    Then, the ranking function for $Q$-tuples $\rank_{\cV,S} \colon \sem{Q}(D)\rightarrow \mathbb{R}, Q(h(\overline{x}))\mapsto \delta_{\cV}(S \cup \{Q(h(\overline{x}))\}))$ is $\overline{y}$-decomposable conditioned on $\overline{z}$ for every $S\subseteq \sem{Q}(D)$.
\end{lemma}
\begin{proof}
    Let $S \subseteq \sem{Q}(D)$.
    Then, we need to show that 
    $$\rank_{\cV,S}(Q(h(\overline{x})))=\delta_{\cV}(S \cup \{Q(h(\overline{x}))\}) = \delta_{\cV}(S)+\mu(\beta(Q(h(\overline{x})))\setminus \bigcup_{s\in S}\beta(s))$$
    is $\overline{y}$-decomposable conditioned on $\overline{z}$.
    Due to Lemma \ref{lem:ydec}, 
    $$\beta_{\bigcup_{s\in S}\beta(s)}\colon Q(h(\overline{x})) \mapsto \beta(Q(h(\overline{x})))\setminus \bigcup_{s\in S}\beta(s)$$
    is $\overline{y}$-decomposable conditioned on $\overline{z}$.
    Then, due to Lemma \ref{lem:mudec}, 
    $$\rank\colon Q(h(\overline{x}))\mapsto \mu(\beta_{\bigcup_{s\in S}\beta(s)}(Q(h(\overline{x}))))$$
    is $\overline{y}$-decomposable conditioned on $\overline{z}$.
    Lastly, due to Lemma \ref{lem:tdec}, 
    $$\rank_{\cV,S} \colon Q(h(\overline{x})) \mapsto \delta_{\cV}(S \cup \{Q(h(\overline{x}))\}) =\delta_{\cV}(S)+\rank(Q(h(\overline{x})))$$
    is $\overline{y}$-decomposable conditioned on $\overline{z}$ as required.
\end{proof}

\begin{lemma}
\label{lem:final}
    Let $(T,\chi)$ be a tree decomposition of the full CQ $Q(\bx)$ and let $\cV=(\cS, \mu, \beta)$ be a volume assignment over $\TUPLES$ such that $\beta$ is compatible with $(T,\chi)$. 
    Then, the ranking function $\rank_{\cV,S} \colon \sem{Q}(D)\rightarrow \mathbb{R}, Q(h(\overline{x}))\mapsto \delta_{\cV}(S \cup \{Q(h(\overline{x}))\}))$ is compatible with $(T,\chi)$ for every $S\subseteq \sem{Q}(D)$.
\end{lemma}
\begin{proof}
    Notice that for every $t\in V(T)$ the ball function $\beta$ is $(\chi(T_t)\setminus\texttt{key}(t))$-decomposable conditioned on $\texttt{key}(t)$.
    Then, due to Lemma \ref{lem:torank}, for any $S\subseteq \sem{Q}(D)$ the rank function $\rank_{\cV,S}$ is $(\chi(T_t)\setminus\texttt{key}(t))$-decomposable conditioned on $\texttt{key}(t)$.
    Thus, in general, for any $S\subseteq \sem{Q}(D)$ the rank function $\rank_{\cV,S}$ is compatible with $(T,\chi)$.
\end{proof}

Next, we recall the theorem proven in \cite{deep2025ranked} (a simplified version)
Note that they assume $\rank$ can be computed in constant time.

\begin{theorem}[\cite{deep2025ranked}]
\label{thm:ranked}
    Let $(T,\chi)$ be a tree decomposition of the full CQ $Q(\bx)$ and $\rank$ a ranking function for $Q$-tuples that is compatible with $(T,\chi)$.
    Then, enumerating $\sem{Q}(D)$ in the order induced by $\rank$ is possible with $O(|D|^{fhw(T,\chi)})$ preprocessing time and $O(\log|D|)$ delay in data complexity.
    In particular, a solution maximizing $\rank$ can be found in time $O(|D|^{fhw(T,\chi)})$.
\end{theorem}

Note that from the proof given in~\cite{deep2025ranked}, it suffices to add the factor $|Q|$ to account for combined complexity.

\thmappeval*

\begin{proof}
Due to Lemma \ref{lem:final}, we can use Theorem \ref{thm:ranked} to solve $\CQNext[\Sigma, \cV]$ in time $O(|Q|\cdot|D|^{fhw(T,\chi)}\cdot T_{\cV})$ (note that we add the factor $T_{\cV}$ as we cannot assume $\rank_{\cV,S}$ as defined in Lemma \ref{lem:final} to be computable in constant time).
Then, due to Theorem \ref{thm:impapprox}, the claim follows.
\end{proof}

\subsection{Proof of Theorem \ref{thm:provenance}}
\label{app:thmprovenance}

\thmprovenance*

\begin{proof}
    Let $Q$ be of the form 
    $$Q(\bx)\leftarrow R_1(\bx_1),\dots,R_m(\bx_m).$$
    Then, for every atom $R_i(\bx_i)$ of $Q$ let us add to the schema the projection onto $\bx$ denoted as $R^{\bx}_i(\bx_i\cap \bx)$ as well as new relations $R^{\bx}_v(\chi(v))$ for nodes $v\in\{v_1,\dots,v_{|V(T_{\bx})|}\}= V(T_{\bx})$.
    Then, let us consider the query
    $$Q^{\bx}(\bx)\leftarrow R^{\bx}_1(\bx_1\cap \bx),\dots,R^{\bx}_m(\bx_m\cap \bx), R^{\bx}_{v_1}(\chi(v_1)), \dots, R^{\bx}_{v_{|V(T_{\bx})|}}(\chi(v_{|V(T_{\bx})|})).$$
    Then, $(T_{\bx}\chi|_{V(T_{\bx})})$ is a tree decomposition of the full CQ $Q^{\bx}$.
    Further, for every $v\in V(T_{\bx})$ let $T_{v}\subseteq T\setminus (T_{\bx}\setminus \{v\})$ be the subtree of $T$ rooted in $v$ that contains $v$ and all its descendants that are not connected to $v$ via $T_{\bx}$.
    Then, we can project all atoms on $\bx_v:=\bigcup_{u\in V(T_v)}\chi(u)$.
    That is, let us add to the schema the relation symbols $R^v_i(\bx_i\cap \bx_v)$ for nodes $v\in\{v_1,\dots,v_{|V(T_{\bx})|}\}= V(T_{\bx})$.
    Then, $(T_v,\chi|_{V(T_v)})$ is a tree decomposition of the (non-full) CQ
    $$Q^v(\chi(v))\leftarrow R^v_1(\bx_1\cap \bx_v), \dots,  R^v_m(\bx_m\cap \bx_v).$$

    Now, our goal is to add data to the database $D$ such that $\sem{Q}(D)$ is apart from the different relations symbol the same as $\sem{Q^{\bx}}(D^{\bx})$.
    To that end, we will intuitively simply add tuples using $R^{\bx}_i$ and $R^v_i$ by projecting the tuples using $R_i$ and tuples $R^{\bx}_v$ that are the results $\sem{Q^v}(D)$.
    That is for every $R_i(h(\bx_i))\in D$ we add to $D^{\bx}$ the tuples
    $$R^{\bx}_i(h(\bx_i\cap \bx)) $$
    and for every $v\in V(T_{\bx})$ the tuple
    $$ R^v_i(h(\bx_i\cap \bx_v)).$$
    Then, further, for every $v\in V(T_{\bx})$ and $Q^v(h(\chi(v)))\in \sem{Q^v}(D^{\bx})$ we add to $D^{\bx}$ the tuple
    $$R^{\bx}_v(h(\chi(v))).$$
    Then, observe that as desired $\sem{Q^{\bx}}(D^{\bx})= \{Q^{\bx}(h(\bx)) \mid Q(h(\bx))\in \sem{Q}(D)\}$.
    Note that this construction can be done in time $O(|D|^{fhw(T,\chi)})$.
    Actually, we would like to annotated the tuple in $D^{\bx}$ such that the annotation of $Q^{\bx}(h(\bx))\in \sem{Q^{\bx}}(D^{\bx})$ is $\betaQD(Q(h(\bx)))$.
    To do so, we annotate $R^{\bx}_i$- and $R^{\bx}_v$-tuples.
    That is, we consider the which-provenance semi-ring \cite{DBLP:journals/vldb/CuiW03} consisting of monomials with no constant and exponents $1$ or $0$.
    That is over variables $\mathcal{Y}$ these are of the form 
    $$\prod_{Y\in \mathcal{Y}}Y^{\alpha_{Y}}$$
    where $\alpha_Y=0$ or $\alpha_Y=1$.
    The addition $\oplus$ and multiplication $\otimes$ are for two non-0 monomials the same, i.e.,
    $$(\prod_{Y\in \mathcal{Y}}Y^{\alpha_{Y}})\oplus (\prod_{Y\in \mathcal{Y}}Y^{\alpha'_{Y}}) = (\prod_{Y\in \mathcal{Y}}Y^{\alpha_{Y}})\otimes (\prod_{Y\in \mathcal{Y}}Y^{\alpha'_{Y}}) = \prod_{Y\in \mathcal{Y}}Y^{\max(\alpha_{Y},\alpha'_Y)}$$
    while $0$ is the additive identity and $1$ is the multiplicative identity.    
    Therefore, 
    $$(\prod_{Y\in \mathcal{Y}}Y^{\alpha_{Y}})\oplus 1 = \prod_{Y\in \mathcal{Y}}Y^{\alpha_{Y}}, \quad (\prod_{Y\in \mathcal{Y}}Y^{\alpha_{Y}})\otimes 0 = 0.$$
    One can also view the elements $\prod_{Y\in \mathcal{Y}}Y^{\alpha_{Y}}$ as subsets of $\mathcal{Y}$ where the operations $\otimes, \oplus$ are then both the union $\cup$.
    Then, we can see $1=\emptyset$ while $0=\bot$ and $\otimes, \oplus$ interact differently with $\emptyset$ and $\bot$.

    We see the tuples $D$ as the variables $\mathcal{Y}$.
    Thus, the domain of the semi-ring consist of sets of tuples, i.e., $2^D\cup\{\bot\}$.
    Then, we annotate the tuples $R^{\bx}_i(h(\bx_i \cap \bx))$ by
    $$\begin{cases}
        \{R_i(h(\bx_i))\} & \text{if } \bx_i \cap \bx = \bx_i,\\
        \emptyset & \text{if } \bx_i \cap \bx \neq \bx_i.
    \end{cases}$$
    and $R^v_i(h(\bx_i\cap \bx_v))$ likewise by 
    $$\begin{cases}
        \{R_i(h(\bx_i))\} & \text{if } \bx_i \cap \bx_v = \bx_i \text{ and } \bx_i\cap \bx \neq \bx_i,\\
        \emptyset & \text{if } \bx_i \cap \bx_v \neq \bx_i \text{ or } \bx_i\cap \bx = \bx_i.
    \end{cases}$$
    We note that by construction for every tuple $R_i(h(\bx_i))\in D$ such that $\bx_i\subseteq \bx$ the tuple $R^{\bx}_i(h(\bx_i))$ is the only tuple annotated by $R_i(h(\bx_i))$.
    On the other hand, for tuples $R_i(h(\bx_i))\in D$ such that $\bx_i\not\subseteq \bx$ there is exactly one $v\in V(T_{\bx})$ such that $\bx_i\subseteq \bx_v$.
    Consequently always exactly one tuple is annotated by $R_i(h(\bx_i))$, i.e., the tuple $R^{\bx}_i(h(\bx_i))$ or $R^v_i(h(\bx_i))$.
    Then, we can compute the annotations of the tuples $Q^v(h(\chi(v)))\in\sem{Q^v}(D^{\bx})$ by evaluating the queries $Q^v$ interpreted as sum-product queries~\cite{DBLP:conf/pods/KhamisNR16, DBLP:journals/jcss/PichlerS13}.
    As operations in the semi-ring might take up to $O(|D|)$ time, this can be done in 
    $$O(|Q|\cdot|D|^{fhw(T_{v},\chi|_{V(T_v)})+1}) = O(|Q|\cdot|D|^{fhw(T,\chi)+1})$$ 
    time for every $v\in V(T_{\bx})$.

    Now, note that tuples $R_i(h(\bx_i))$ where $\bx_i\not\subseteq \bx$ are part of the annotation of a $R^{\bx}_{v}$-tuple when it contributed to that tuple being in $D^{\bx}$, i.e., when it is part of its provenance.
    However, still, every (now not only for $\bx_i\not\subseteq \bx$) tuples $R_i(h(\bx_i))$ is part of the annotations of at most one tuple type used in query $Q^{\bx}$, i.e., either of a $R^{\bx}_i$- or a $R^{\bx}_v$-tuple type.
    We say that we can uniquely assign $R_i$ to an atom of $Q^{\bx}$

    Further, we could interpret $Q^{\bx}$ as a sum-product query and annotate $Q^{\bx}(h(\bx))\in \sem{Q^{\bx}}(D^{\bx})$ analogously.
    Then, note that with this construction and interpreting the monomials as sets of tuples, we get $\betaQD(Q(h(\bx)))$ for $Q(h(\bx))\in \sem{Q}(D)$ is equal to the annotation of $Q^{\bx}(h(\bx))\in\sem{Q^{\bx}}(D^{\bx})$.
    We extend $\betaQD$ to $Q^{\bx}$-tuples and define $\betaQD(Q^{\bx}(h(\bx))) = \betaQD(Q(h(\bx)))$.
    Then, for a fixed set of answers $\{Q(h_i(\bx))\mid i\}\subseteq \sem{Q}(D)$ and arbitrary single $Q(h(\bx))\in \sem{Q}(D)$, computing the marginal diversity is very easy: 
    We simply compute once $\betaQD(\{Q^{\bx}(h_i(\bx))\mid i\}):=\bigcup_{i}\betaQD(Q^{\bx}(h_i(\bx)))$, i.e., we take the union of the annotations in $\sem{Q^{\bx}}(D^{\bx})$,
    and the marginal diversity is then
    $$|\betaQD(Q^{\bx}(h(\bx)))\setminus \betaQD(\{Q^{\bx}(h_i(\bx))\mid i\})|.$$
    Thus, this is doable in $T_{\VQD}=O(|D|)$ time.

    Now, we aim to prove that $\betaQD$ is compatible with $(T_{\bx},\chi|_{V(T_{\bx})})$.
    If we manage to do that, we can apply Theorem \ref{thm:appeval} to complete the proof.
    To that end, denote $T_{v,\bx}$ the subtree of $T_{\bx}$ rooted in $v\in V(T_{\bx})$.
    Thus, we have to show for every $v\in V(T_{\bx})$ that $\betaQD$ is $(\chi(T_{v,\bx})\setminus \texttt{key}(v))$-decomposable conditioned on $\texttt{key}(v)$.

    Note that for any atom of $Q^{\bx}$, the atom is either covered by $\texttt{key}(t)$, by $T_{v,\bx}$ or by $T_{\bx}-T_{v,\bx}$.
    With this we mean that for atoms $R^{\bx}_i(\bx_i\cap \bx)$ either:
    \begin{enumerate}
        \item $\bx_i\cap \bx\subseteq \texttt{key}(v)$,
        \item $\bx_i\cap \bx\subseteq \bigcup_{u\in V(T_{v,\bx})}\chi(u)$ and $\bx_i\cap \bx\not \subseteq \texttt{key}(v)$, or \label{app:type}
        \item $\bx_i\cap \bx\subseteq \bigcup_{u\in V(T_{\bx})\setminus V(T_{v,\bx})}\chi(u)$ and $\bx_i\cap \bx\not \subseteq \texttt{key}(v)$.
    \end{enumerate}
    We call this the type of $R^{\bx}_i$ below.
    The same holds for $R^{\bx}_u(\chi(u))$ with $\chi(u)$ used instead of $\bx_i\cap \bx$.
    Then, let $f$ be a homomorphism over $\texttt{key}(v)$, $h,h'$ be homomorphisms over $\chi(T_{v,\bx})\setminus \texttt{key}(v)$, and $g$ be a homomorphism over $\bx \setminus \chi(T_{v,\bx})$.
    Then, 
    \begin{align}
        \betaQD(Q^{\bx}((f\cup h \cup g)(\bx)))\setminus \betaQD(Q^{\bx}((f\cup h' \cup g)(\bx))) =: \betaQD(h,h')     \label{app:betahh}       
    \end{align}
    are the $R_i$-tuples that contribute to $Q((f\cup h \cup g)(\bx))$ being in $\sem{Q}(D)$ and where $\bx_i\cap \bx\subseteq \chi(T_{v,\bx})$ but also $\bx_i\cap \bx\not \subseteq \texttt{key}(v)$.
    Put differently, these stem from atoms $R^{\bx}_i(\bx_i\cap \bx), R^{\bx}_u(\chi(u))$ of type $\eqref{app:type}$.
    Then, as we can uniquely assign every $R_i$ to an atom of $Q^{\bx}$, this value $\betaQD(h,h')$ given in Equation \eqref{app:betahh} is independent of the choice of $g$.
    Consequently, $\betaQD$ is compatible with $(T_{\bx},\chi|_{V(T_{\bx})})$.

    Now, applying Theorem \ref{thm:appeval}, we can $(1-\nicefrac{1}{e})$-approximate a $k$-diversity set $\{Q^{\bx}(h_i(\bx))\mid i\}\subseteq \sem{Q^{\bx}}(D^{\bx})$ in time $$O(|Q|\cdot |D|^{fhw(T_{\bx},\chi|_{V(T_{\bx})})+1}\cdot k)=O(|Q|\cdot |D|^{fhw(T,\chi)+1}\cdot k).$$
    Then, as $\delta_{\VQD}(\{Q^{\bx}(h_i(\bx))\mid i\})=\delta_{\VQD}(\{Q(h_i(\bx))\mid i\})$ for any subsets of $\sem{Q^{\bx}}(D^{\bx})$ and $\sem{Q}(D)$, respectively.
    Thus, $S:=\{Q(h_i(\bx))\mid i\}\subseteq \sem{Q}(D)$ is a $(1-\nicefrac{1}{e})$-approximate $k$-diversity set.
    The total time used is $O(|Q|\cdot |D|^{fhw(T,\chi)+1}\cdot k)$ as claimed.
\end{proof}

\end{document}